\documentclass[11pt]{article}
\usepackage[usenames,dvipsnames,svgnames,table]{xcolor} % usenames : 16 basic colors,  dvipsnames : other 68 colors, svgnames : another 150 colors. It is used instead of \usepackage{color}. Can find descriptions for color in xcolor package. 

\usepackage{enumitem} % for \begin{description}
\usepackage{rotfloat} % For [H] option in sidewaystable
\usepackage{graphicx}
\usepackage{epsfig}
\usepackage{amssymb, amsmath}
\usepackage{verbatim}
\usepackage{natbib}
\usepackage{authblk}
\usepackage{kotex}
\usepackage{multirow}
\usepackage[colorlinks]{hyperref} % it is needed for todonotes package. 
\usepackage[colorinlistoftodos, textsize=scriptsize]{todonotes} % For making notes. disable option removes all notes. 
\usepackage{subcaption} % for \begin{subtable}

\newtheorem{theorem}{Theorem}[section]

\newtheorem{corollary}[theorem]{Corollary}

\newenvironment{proof}[1][Proof]{\begin{trivlist}
		\item[\hskip \labelsep {\bfseries #1}]}{\end{trivlist}}

\newenvironment{remark}[1][Remark]{\begin{trivlist}
		\item[\hskip \labelsep {\bfseries #1}]}{\end{trivlist}}

\DeclareMathOperator*{\argmax}{argmax}

\newcommand{\bea}{\begin{eqnarray*}}
	\newcommand{\eea}{\end{eqnarray*}}
\newcommand{\bean}{\begin{eqnarray}}
	\newcommand{\eean}{\end{eqnarray}}

\newcommand{\bfX}{{\bf X}}
\newcommand{\bfY}{{\bf Y}}

\newcommand{\sg}{\Sigma}
\newcommand{\what}{\widehat}

\newcommand{\calC}{\mathcal{C}}

\newcommand{\bbP}{\mathbb{P}} 
\newcommand{\bbR}{\mathbb{R}}
\newcommand{\bbE}{\mathbb{E}}

\begin{document}

\title{Bayesian Bandwidth Test and Selection for High-dimensional Banded Precision Matrices}
\author[1]{Kyoungjae Lee}
\author[1]{Lizhen Lin}
\affil[1]{Department of Applied and Computational Mathematics and Statistics, The University of Notre Dame}

\maketitle

\begin{abstract}
Assuming a banded structure is one of the common practice in the estimation of high-dimensional precision matrix. 
In this case, estimating the bandwidth of the precision matrix is a crucial initial step for subsequent analysis.
Although there exist some consistent frequentist tests for the bandwidth parameter, bandwidth selection consistency for precision matrices has not been established in a Bayesian framework.
In this paper, we propose a prior distribution tailored to the bandwidth estimation of high-dimensional precision matrices.
The banded structure is imposed via the Cholesky factor from the modified Cholesky decomposition.
We establish the strong model selection consistency for the bandwidth as well as the consistency of the Bayes factor. The convergence rates for Bayes factors under both the null and alternative hypotheses are derived which yield similar order of rates. As a by-product, we also proposed an estimation procedure for the Cholesky factors yielding an almost optimal order of convergence rates. 
Two-sample bandwidth test is also considered, and it turns out that our method is able to consistently detect the equality of bandwidths between two precision matrices.
The simulation study confirms that our method in general outperforms the existing frequentist and Bayesian methods.
\end{abstract}

Key words: Precision matrix; Bandwidth selection; Cholesky factor; Convergence rates of Bayes factor.

\section{Introduction}

Estimating a large covariance or precision matrix is a challenging task in both frequentist and Bayesian frameworks.
When the number of variables $p$ is larger than the sample size $n$, the traditional sample covariance matrix does not provide a consistent estimate of the true covariance matrix \citep{johnstone2009consistency}, and the inverse Wishart prior leads to the posterior inconsistency \citep{lee2018optimal}.
To overcome this issue, various restricted classes of matrices have been investigated such as the bandable matrices \citep{bickel2008regularized,cai2010optimal,hu2017minimax,banerjee2014posterior,lee2017estimating}, sparse matrices \citep{cai2012minimax,banerjee2015bayesian,xiang2015high,cao2017posterior} and low-dimensional structural matrices \citep{fan2008high,cai2015optimal,pati2014posterior,gao2015rate}.
In this paper, we focus on \emph{banded precision matrices}, where the banded structure is encoded via the Cholesky factor of the precision matrix.
We are in particular interested in the estimation of the bandwidth parameter and construction of Bayesian bandwidth tests for one or two banded precision matrices.
Inference of the bandwidth is of great importance for detecting the dependence structure of the ordered data.
Moreover, it is a crucial initial step for subsequent analysis such as linear or quadratic discriminant analysis.

Bandwidth selection of the high-dimensional precision matrices has received increasing attention in recent years.
\cite{an2014hypothesis} proposed a test for  bandwidth selection, which is asymptotically normal under the null hypothesis and has a power tending to one.
Based on the proposed test statistics, they constructed a backward procedure to detect the true bandwidth by controlling the familywise errors.
%Their procedures can detect the true bandwidth under mild conditions.
%but, to study the power of the proposed test, they assumed the nonzero elements of Cholesky factor tend to zero at some rate.
\cite{cheng2017test} suggested a bandwidth test without assuming any specific parametric distribution for the data and obtained a  result similar to that of \cite{an2014hypothesis}.
%However, their test statistic requires a calculation of a consistent estimator for the precision matrix.

In the Bayesian literature, \cite{banerjee2014posterior} studied the estimation of bandable precision matrices which include the banded precision matrix as a special case.
They derived the posterior convergence rate of the precision matrix under the $G$-Wishart prior \citep{roverato2000cholesky}.
\cite{lee2017estimating} considered a similar class to that of \cite{banerjee2014posterior}, but assumed  bandable Cholesky factors instead of bandable precision matrices.
They showed the posterior convergence rates of the precision matrix as well as the minimax lower bounds.
In both  works,  posterior convergence rates were obtained for a given (fixed) bandwidth, and the posterior mode was suggested as a bandwidth estimator in practice.
However, no theoretical guarantee is provided for such  estimators. Further, no Bayesian bandwidth test exists  for one- or two-sample problems. 
%such as selection consistency is provided for these estimators.

This gap in the literature motivates us to investigate theoretical properties related to the general problem of  bandwidth test and selection, and propose estimators or tests with theoretical guarantees.
In this paper, we use the modified Cholesky decomposition of the precision matrix and assume  banded Cholesky factors.
The induced precision matrix also has  banded structure.
The key difference from \cite{lee2017estimating} is on the choice of prior distributions which will be introduced in Section \ref{subsec:prior}.
In addition, we focus on bandwidth selection and tests, while \cite{lee2017estimating} mainly studied the convergence rates of the precision matrix for a given or fixed bandwidth.

There are two main contributions of this paper.
First, we suggest a Bayesian procedure for banded precision matrices and prove the bandwidth selection consistency (Theorem \ref{theorem:band_sel}) and consistency of the Bayes factor (Theorem \ref{theorem:one_band}).
To the best of our knowledge, our work is the first that has established the bandwidth selection consistency for precision matrices under a Bayesian framework, which implies that the marginal posterior probability for the true bandwidth tends to one as $n\to\infty$.%\cite{lee2017minimax} proved the strong model selection consistency for the sparse DAG models, but they used the fractional likelihood approach, which cannot be applicable to the Bayesian testing. 
\cite{cao2017posterior} proved  strong model selection consistency for the sparse directed acyclic graph models, but their method is not applicable to the bandwidth selection problem since it is not adaptive to the unknown sparsity. 
%Furthermore, it requires strong conditions in terms of the dimensionality, true sparsity, eigenvalues of the true precision matrix and smallest magnitude of the nonzero elements in the Cholesky factor.
%We will discuss this later in more details.
Second, we also prove the consistency of the Bayes factor for two-sample bandwidth testing problem (Theorem \ref{theorem:two_band}) and derived the convergence rates of the Bayes factor under both the null and alternative hypotheses.  
Our method  is able to consistently detect the equality of bandwidths between two different precision matrices.
%Furthermore, we calculate the convergence rates of the Bayes factors for both one sample and two sample bandwidth testing problems.
To the best of our knowledge, this is also the first consistent two-sample bandwidth test result in both frequentist and Bayesian literature.
The existing literature (frequentist) focused only on the one-sample bandwidth testing \citep{an2014hypothesis,cheng2017test}. 
%and it is unclear whether these methods can be extended to the two sample testing problem.

%To the best of our knowledge, these are novel results in Bayesian asymptotics tailored to the bandwidth selection problem.

The rest of the paper is organized as follows.
Section \ref{sec:prelim} introduces the notations, model, priors and assumptions used.
Section \ref{sec:main} describes main results of this paper: bandwidth selection consistency and convergence rates of one-  and two-sample bandwidth tests.
Simulation study and real data analysis are presented in Section \ref{sec:numerical} to show the practical performance of the proposed method.
In Section \ref{sec:discussion}, concluding remarks and topics for the future work are given. The appendix includes a result on the nearly optimal estimation of the Cholesky factors, and  proofs of main results.

%%%%%%%%%%%%%%%%%%%%%%%%%%%%%%%%%%%%%%%%%%%%%
\section{Preliminaries}\label{sec:prelim}

\subsection{Notations}\label{subsec:notation}

For any real numbers $a$ and $b$, we denote $a \wedge b$ and $a \vee b$ as the minimum and maximum of $a$ and $b$, respectively.
For any sequences $a_n$ and $b_n$, we denote $a_n = o(b_n)$ if $a_n / b_n \to 0$ as $n\to \infty$.
We write $a_n \lesssim b_n$, or $a_n = O(b_n)$, if there exists an universal constant $C>0$ such that $a_n \le C b_n$ for any $n$.
%For any random variable $Z_n$ with a distribution function $P$, we denote $Z_n = O_p(1)$ if, for any $\epsilon>0$, there exist a constant $C>0$ and an integer $N>0$ such that $P(|Z_n| >C) \le \epsilon$ for any $n \ge N$.
We define vector $\ell_2$- and $\ell_\infty$-norms as $\|a\|_2 = (\sum_{j=1}^p a_j^2 )^{1/2}$ and $\|a\|_\infty= \max_{1\le j \le p}|a_j|$ for any $a = (a_1,\ldots,a_p)^T \in \bbR^p$.
For a matrix $A$, the matrix $\ell_\infty$-norm is defined as $\|A\|_\infty = \sup_{\|x\|_\infty=1 }\|Ax\|_\infty$.
We denote $\lambda_{\min}(A)$ and $\lambda_{\max}(A)$ as the minimum and maximum eigenvalues of $A$, respectively.

\subsection{Gaussian Models}\label{subsec:gaussian}

%When the random variables have a natural ordering, one common approach for the estimation of high-dimensional covariance (or precision) matrices is to adopt banded structures. 
%%The banded structure can be encoded in many different ways.
%One popular model to incorporate banded structure is a Gaussian directed acyclic graph (DAG) model in which the bandwidth can be encoded by the Cholesky factor via the modified Cholesky decomposition (MCD) below. 

%A directed graph ${\cal{D}}= (V,E)$ consists of vertices $V=\{1,\ldots, p\}$ and directed edges $E$.
%%{\color{red} An edge consists of two different edges, and it is called directed edge if it has a direction.}
%For any $i, j \in V$, we denote $(i,j)\in E$ as a directed edge $i \to j$ and call $i$ a {\it parent} of $j$.
%A DAG is a directed graph with no directed cycle.
%In this paper, we assume a {\it parent ordering} is known, where $i<j$ holds for any parent $i$ of $j$ in a DAG ${\cal{D}}$, which has been commonly used in the literature.
%%It has been commonly used in the literature including \cite{an2014hypothesis}, \cite{cheng2017test}, \cite{khare2016convex} and \cite{cao2017posterior}.
%A multivariate normal distribution $N_p(0,\Omega^{-1})$ is said to be a {\it Gaussian DAG model over} ${\cal{D}}$, if $Y=(Y_1,\ldots,Y_p)^T \sim N_p(0,\Omega_n^{-1})$ satisfies
%\bea
%Y_i \perp \{Y_j \}_{j<i, j \in pa_i({\cal{D}})} \,\,\big| \,\,  \{ Y_j \}_{j\in pa_i({\cal{D}}) }
%\eea
%for any $i=1,\ldots,p$, where $pa_i({\cal{D}})$ is the set of all parents of $i$.

We consider a Gaussian model
\bean\label{model}
X_1,\ldots,X_n \mid \Omega_n &\overset{i.i.d.}{\sim}& N_p(0, \Omega_n^{-1}),
\eean
where $\Omega_n = \sg_n^{-1}$ is a $p\times p$ precision matrix and $X_i = (X_{i1},\ldots, X_{ip})^T\in \bbR^p$ for all $i=1,\ldots,n$. 
For any positive definite matrix $\Omega_n$, there exist unique lower triangular matrix $A_n=(a_{jl})$ and diagonal matrix $D_n=diag(d_j)$  such that
\bean\label{chol}
\Omega_n &=& (I_p - A_n)^T D_n^{-1} (I_p - A_n) ,
\eean
where $a_{jj} =0$ and $d_j >0$ for all $j=1,\ldots, p$, by the modified Cholesky decomposition (MCD). 
We call $A_n$ the {\it Cholesky factor}.
%It can be easily shown that $a_{jl} \neq 0$ if and only if $l \in pa_j({\cal{D}})$, so the Cholesky factor $A_n$ uniquely determines a DAG ${\cal{D}}$.
Define $k$ as the {\it bandwidth} of a matrix if the off-diagonal elements of the matrix farther than $k$ from the diagonal are all zero.
If the bandwidth of the Cholesky factor is $k$,   model \eqref{model} can be represented as
\bean\label{model2}
\begin{split}
	{X}_{i1} \mid d_1 \,\,&\overset{i.i.d.}{\sim}\,\, N(0,\, d_1 ), \\
	{X}_{ij} \mid a_{j}^{(k)}, d_j, k \,\,&\overset{ind}{\sim}\,\,  N \Big( \sum_{l = (j-k)_1}^{j-1} {X}_{il} a_{jl} ,\, d_j \Big) ,~ ~j=2,\ldots,p 
\end{split}
\eean
for all $i=1,\ldots,n$, where $a_{j}^{(k)} = (a_{j l})_{ (j-k)_1 \le l \le j-1} \in \bbR^{k_j}$, $(j-k)_1 = 1\vee (j-k)$ and $k_j = k \wedge(j-1)$. 
The above representation enables us to adopt priors and techniques in the linear regression literature.

We are interested in the consistent estimation and hypothesis test of the bandwidth $k$ of the precision matrix.
From the decomposition \eqref{chol}, \emph{the bandwidth of $A_n$ is $k$ if and only if the bandwidth of $\Omega_n$ is $k$.}
Thus, we can infer the bandwidth of the precision matrix by inferring that of the Cholesky factor.

\subsection{Prior Distribution}\label{subsec:prior}

Let $\tilde{X}_j \in \bbR^n$ and ${\bfX_{j (k)}} \in \bbR^{n \times k_j}$ be sub-matrices consisting of $j$th and $ (j-k)_1 ,\ldots, (j-1)$th columns of $\bfX_n =( X_1^T,\ldots, X_n^T)^T \in \bbR^{n \times p}$, respectively.
We suggest the following prior distribution
\begin{eqnarray}
a_{j}^{(k)} \mid d_j, k &\overset{ind}{\sim}& N_{k_j} \left( \what{a}_{j}^{(k)},\,\,   \frac{d_j}{\gamma } \big( {\bfX_{j (k)}}^T  {\bfX_{j (k)}}\big)^{-1}  \right) ,~~ j=2,\ldots,p , \label{prior1} \\
\pi(d_{j})  &\overset{i.i.d.}{\propto}& d_j^{ \tau n/2 -1 }, ~~ j=1,\ldots,p , \label{prior2} \\
k &\sim& \pi(k) , 
~~ k=0,1,\ldots,R_n \label{prior3}
\end{eqnarray}
for some positive constants $\gamma, \tau$ and positive sequence $R_n$, where $\what{a}_j^{(k)} =  ({\bfX_{j (k)}}^T {\bfX_{j (k)}})^{-1} {\bfX_{j (k)}}^T \tilde{X}_j$.
The conditional prior distribution for $a_j^{(k)}$ is a version of the Zellner's $g$-prior \citep{zellner1986assessing, martin2017empirical} in the linear regression literature.
Note that model \eqref{model2} is equivalent to $\tilde{X}_j \mid a_j^{(k)},d_j,k \sim N_n \big( \bfX_{j (k)} a_j^{(k)}, d_j I_n \big)$.
Due to the conjugacy, it enables us to calculate the posterior distribution in a closed form up to some normalizing constant.
The prior for $d_j$ is carefully chosen to reduce the posterior mass towards large bandwidth $k$.
We emphasize here that one can use the usual non-informative prior $\pi(d_j) \propto d_j^{-1}$, but necessary conditions for the main results in Section \ref{sec:main} should be changed.
This issue will be discussed in more details in the next paragraph.
We assume the prior $\pi(k)$ to have the support on $0,1,\ldots,R_n$.
We will introduce condition (A4) for $\pi(k)$ and the hyperparameters in Section \ref{subsec:assump}, and show that $\pi(k)\propto 1$ is enough to establish the main results in Section \ref{sec:main}.

The priors \eqref{prior1}--\eqref{prior3} lead to the following joint posterior distribution,
\begin{eqnarray}\label{post}
\begin{split}
a_{j}^{(k)} \mid d_{j}, k, \bfX_n \,\,&\overset{ind}{\sim}\,\, N_{k_j} \left(  \what{a}_{j}^{(k)}   , ~\frac{d_{j}}{1+\gamma}\big({\bfX_{j (k)}}^T  {\bfX_{j (k)}}\big)^{-1}  \right), \quad j=2,\ldots,p, \\
d_{j} \mid k, \bfX_n \,\,&\overset{ind}{\sim} \,\, IG \left(\frac{(1-\tau) n}{2}, \frac{n}{2}  \what{d}_{j}^{(k)} \right) , \quad j=1,\ldots, p ,\\
\pi(k \mid \bfX_n) \,\,&\propto\,\, \pi(k) \prod_{j=2}^p \left(1+ \frac{1}{\gamma} \right)^{-\frac{k_j}{2}} (\what{d}_{j}^{(k)})^{- \frac{(1-\tau) n}{2}}, \quad k=0,1,\ldots,R_n ,
\end{split}
\end{eqnarray}
provided that $\tau <1$, where $\what{d}_{j}^{(k)} = \tilde{X}_j^T ( I - \tilde{P}_{jk}) \tilde{X}_j /n$ and $\tilde{P}_{jk} = {\bfX_{j (k)}} ({\bfX_{j (k)}}^T {\bfX_{j (k)}})^{-1} {\bfX_{j (k)}}^T$.
The marginal posterior $\pi(k\mid\bfX_n)$ consists of two parts: the penalty on the model size, $\pi(k)\prod_{j=2}^p (1+ 1/\gamma)^{-k_j/2}$, and the estimated residual variances, $\prod_{j=2}^p (\what{d}_j^{(k)})^{-(1-\tau)n/2}$.
Thus,  priors \eqref{prior1} and \eqref{prior2} naturally impose the penalty term $\prod_{j=2}^p (1+ 1/\gamma)^{-k_j/2}$ for the marginal posterior $\pi(k\mid \bfX_n)$.

The effect of prior $\pi(d_j)\propto d_j^{\tau n/2 -1}$ appears in marginal posterior for $k$.
Compared with the prior $\pi(d_j)\propto d_j^{-1}$, it produces the term $(\what{d}_j^{(k)})^{-(1-\tau)n/2}$ instead of $(\what{d}_j^{(k)})^{-n/2}$.
Thus, it reduces the posterior mass towards large bandwidth $k$ since $\what{d}_j^{(k)}$ decreases as $k$ grows.
We conjecture that, at least for our prior choice of $\pi(a_j^{(k)}\mid d_j,k)$ with a constant $\gamma>0$, this power adjustment of $\what{d}_j^{(k)}$ is essential to prove the selection consistency for $k$.
Suppose  we use the prior $\pi(d_j)\propto d_j^{-1}$.
Similar to the proof of Theorem \ref{theorem:band_sel}, to obtain the selection consistency, we will use the inequality 
\bean \label{key_ineq}
\pi(k\mid\bfX_n) &\le& \frac{\pi(k\mid\bfX_n)}{\pi(k_0\mid\bfX_n)}  \,=\, \frac{\pi(k)}{\pi(k_0)} \prod_{j=2}^p \left(1+ \frac{1}{\gamma} \right)^{-\frac{k_j - k_{0j}}{2}} \left(\frac{\what{d}_{j}^{(k)} }{\what{d}_{j}^{(k_0)}} \right)^{-\frac{n}{2}},
\eean
and show that the expectation of the right hand side term converges to zero for any $k\neq k_0$ as $n\to\infty$, where $k_0$ is the true bandwidth. 
Note that unless $\pi(k_0\mid\bfX_n)$ shrinks to zero, the inequality causes only a constant multiplication.
The most important task is dealing with the last term in \eqref{key_ineq}, $(\what{d}_j^{(k)}/ \what{d}_{j}^{(k_0)} )^{-n/2}$.
Concentration inequalities for chi-square random variables (for examples, see Lemma 3 in \cite{yang2016computational} and Lemma 4 in \cite{shin2018scalable}) suggest an upper bound $p^{\alpha(k_j - k_{0j})}$ with high probability for any $2\le j\le p$, $k>k_0$ and some constant $\alpha>0$.
In this case, the hyperparameter $\gamma$ should be of order $p^{-\alpha'}$ for some constant $\alpha'> 2\alpha$ to make the right hand side  in \eqref{key_ineq} converge to zero.
Then, with the choice $\gamma \asymp p^{-\alpha'}$, condition (A2), which will be introduced in Section \ref{subsec:assump}, should be modified by replacing $1/n$ with $(\log p) /n$ to achieve the selection consistency.
In summary, the main results in this paper still hold for the prior $\pi(d_j) \propto d_j^{-1}$, but it requires stronger conditions due to technical reasons.
We state the results using prior \eqref{prior2} to emphasize that the bandwidth selection problem essentially requires weaker condition than the usual model selection problem.

\begin{remark}
	If we adopt the fractional likelihood  \citep{martin2017empirical}, we can achieve the selection consistency (Theorem \ref{theorem:band_sel}) with the prior $\pi(d_j)\propto d_j^{-1}$ instead of \eqref{prior2} under similar conditions in Theorem \ref{theorem:band_sel}.
	However, with the fractional likelihood, we cannot calculate the Bayes factor which is essential to describe the Bayesian test results in Sections \ref{subsec:tests} and \ref{subsec:two_tests}.
\end{remark}

\begin{remark}
	There are two consequences by using the data-dependent mean $\what{a}_{j}^{(k)}$.  
	First, we can avoid assuming an upper bound condition for $\| \bfX_{j(k_0)}a_{0,j}^{(k_0)} \|_2$  or $\|a_{0,j}^{(k_0)}\|_2$, where $a_{0,j}^{(k_0)} = (a_{0,jl} )_{(j-k_0)_1\le l\le j-1 }$ denotes the sub-vector of the true Cholesky factor.
	An upper bound condition is required if we adopt the Zellner's $g$-prior with zero mean \citep{shang2011consistency}, e.g., \cite{yang2016computational} assumed $\| \bfX_{j(k_0)}a_{0,j}^{(k_0)} \|_2^2 \le \gamma^{-1} d_{0j}\log p$ in order to prove selection consistency for the regression coefficient vector.
	Second, we do not need to assume the so-called information paradox of Zellner's $g$-prior \citep{liang2008mixtures}, which corresponds to $\gamma = p^{-2c}$ for some $c \ge 1/2$ in our notation.
	In this paper, we assume $\gamma$ is a constant satisfying some conditions in Section \ref{subsec:assump}.
\end{remark}

\subsection{Assumptions}\label{subsec:assump}

%Throughout the paper, we consider the high-dimensional setting where the number of variables $p$ increases to the infinity as $n\to \infty$.
We denote $\Omega_{0n}$ as the true precision matrix whose MCD is given by $\Omega_{0n}=(I_p - A_{0n})^T D_{0n}^{-1} (I_p - A_{0n})$.
Let $\bbP_0$ and $\bbE_0$ be the probability measure and expectation corresponding to model \eqref{model} with $\Omega_{0n}$.
For the true Cholesky factor $A_{0n} = ( a_{0, jl} )$, we denote $k_0$ as the true bandwidth. 
We introduce conditions (A1)--(A4) for the true precision matrix and  priors \eqref{prior1}--\eqref{prior3}:\\
{\bf (A1)}\label{A1}.  Assume that $p$ increases to the infinity as $n\to \infty$.
Furthermore, there exist positive sequences $\epsilon_{0n} \le 1$ and $\zeta_{0n} \ge 1$ such that
$\epsilon_{0n} \le \lambda_{\min}(\Omega_{0n})\le \lambda_{\max}(\Omega_{0n})\le \zeta_{0n}$ for every $n\ge 1$ and $\zeta_{0n} \log p/\epsilon_{0n} = o(n)$.   \\
{\bf  (A2)}\label{A2}. For a given positive constant $\tau \in (0, 0.4]$ in  prior \eqref{prior2}, there exists a positive constant $M_{\rm bm}$ such that for every $n\ge 1$,
\bea
\min_{j,l: a_{0,jl}\neq 0} |a_{0,jl}|^2 &\ge& \frac{10 M_{\rm bm} }{ \tau(1-\tau) \, n } \, \frac{\zeta_{0n} }{ \epsilon_{0n}} .
\eea
{\bf  (A3)}\label{A3}.   The sequence $R_n$ in  prior \eqref{prior3} satisfies  $k_0 \le R_n \le \min \big\{ n\tau \epsilon_1 (1+\epsilon_1)^{-1} , (1-\epsilon_2) p \big\} $ for some small $0<\epsilon_1, \epsilon_2 <1$ and all sufficiently large $n$.   \\
{\bf  (A4)}\label{A4}. For given positive constants $\gamma$ and $\tau \in (0, 0.4]$ in priors \eqref{prior1} and \eqref{prior2}, assume that
\bean
\sum_{k>k_0}  \frac{\pi(k)}{\pi(k_0)} \left\{ C_{\gamma,\tau} \right\}^{- (k-k_0)\cdot (p - \frac{k+k_0+1}{2} ) } &=& o(1) ,  \label{size_cond} \\
\sum_{k <k_0} \frac{\pi(k)}{\pi(k_0)}  \left\{ C_{\gamma, M_{\rm bm}} \right\}^{ (k_{0}-k)\cdot (p - \frac{k+k_0+1}{2} ) } &=& o(1), \label{size_cond2}
\eean
where $C_{\gamma,\tau} = \big\{ (1+\gamma^{-1}) \cdot \tau (1+\epsilon_1)^{-1} \big\}^{1/2}$ and $C_{\gamma, M_{\rm bm}} = 2\, (1+\gamma^{-1} )^{1/2} \exp(-M_{\rm bm} )$.

Now, let us describe the above conditions in more detail.
The bounded eigenvalue condition for the true precision matrix is common in the high-dimensional precision matrix literature \citep{banerjee2014posterior,banerjee2015bayesian,xiang2015high,ren2015asymptotic}.
We allow that $\epsilon_{0n} \to 0$ and $\zeta_{0n} \to \infty$ as $n\to\infty$, so condition (A1) is much weaker than the condition in the above literature, which assumes $\epsilon_{0n}=\zeta_{0n}^{-1}= \epsilon_0$ for some small constant $\epsilon_0>0$.
\cite{cao2017posterior} also allowed diverging bounds, but assumed that $\zeta_{0n} = \epsilon_{0n}^{-1}$ and $(\log p/n)^{1/2 - 1/(2+t)} = o (\epsilon_{0n}^4 )$ for some $t>0$.
If we assume that $\zeta_{0n} = \epsilon_{0n}^{-1}$, then  condition (A1) implies that $\log p /n = o(\epsilon_{0n}^2)$, which is much weaker than the condition used in \cite{cao2017posterior}.

Condition (A2) is called the {\it beta-min condition}. 
If we assume that $\epsilon_{0n} = O(1)$ and $\zeta_{0n}=O(1)$, in our model it only requires the lower bound of the nonzero elements to be of order $O(1/\sqrt{n})$. 
In the sparse regression coefficient literature, the lower bound of the nonzero coefficients is usually assumed to be $\sqrt{\log p/n}$ up to some constant \citep{castillo2015bayesian,yang2016computational,martin2017empirical}.
Here, the $\sqrt{\log p}$ term can be interpreted as a price coming from the absence of information on the zero-pattern.
Condition (A2) reveals the fact that, under the banded assumption, we do not need to pay this price anymore.

Condition (A3) ensures that the true bandwidth $k_0$ lies in the support of $\pi(k)$.
Note that $k_0 \le (1-\epsilon_2)p$ is not an additional condition because the support of bandwidth should be smaller than $p$. 
The condition $k_0 \le n\tau \epsilon_1(1+\epsilon_1)^{-1}$ is needed for the selection consistency, which holds if we choose $R_n = C \vee \{n\tau \epsilon_1(1+\epsilon_1)^{-1} \}$ for some large constant $C>0$. 
Although this is slightly stronger than the condition $k_0 \le n-4$ in \cite{an2014hypothesis}, it is much weaker than those in other works.
For examples, \cite{banerjee2014posterior} assumed $k_0^5 = o(n/\log p)$ for the consistent estimation of precision matrix, and \cite{cheng2017test} assumed $k_0 = O( [n/\log p]^{1/2} )$ for  theoretical properties.

The equations \eqref{size_cond} and \eqref{size_cond2} in  condition (A4) guarantee $\bbE_0 [\pi(k>k_0\mid \bfX_n)] = o(1)$ and $\bbE_0 [\pi(k <k_0\mid \bfX_n)] = o(1)$, respectively.
Here we give some examples for $\pi(k)$ satisfying  conditions \eqref{size_cond} and \eqref{size_cond2}: if we choose
\bean\label{pi_k_xi}
\pi(k) &\propto& \xi^{k(p - \frac{k+1}{2})}
\eean
with $C_{\gamma,\tau}^{-1} < \xi < C_{\gamma, M_{\rm bm}}^{-1}$, it satisfies the conditions.
Furthermore, if we choose $\xi=1$, which leads to
\bean\label{pi_k_const}
\pi(k) &=& \frac{1}{R_n+1},
\eean
the conditions are met if $\tau >  (1+\epsilon_1)(1+\gamma^{-1})^{-1}$ and $\exp(M_{\rm bm}) >2 (1+\gamma^{-1})^{1/2}$.

\begin{remark}
	In the sparse linear regression literature, a common choice for the prior on the unknown sparsity $k$ is $\pi(k) \propto p^{-c k}$
	for some constant $c>0$. 
	See \cite{castillo2015bayesian}, \cite{yang2016computational} and \cite{martin2017empirical}.
	If we adopt this type of the prior into the bandwidth selection problem, a naive approach is using $\pi(k) \propto p^{-c k}$ for each row of the Cholesky factor: it results in $\pi(k) \propto p^{-c k(p-k)}$. 
	To obtain the strong model selection consistency, in this case, $M_{\rm bm}$ in condition (A2) has to be $M_{\rm bm} = M_{\rm bm}' \, \log p$ for some constant $M_{\rm bm}'>0$. 
	Thus, it unnecessarily requires stronger beta-min condition, which can be avoided by using $\pi(k)$ like \eqref{pi_k_xi} or \eqref{pi_k_const}.
\end{remark}

%%%%%%%%%%%%%%%%%%%%%%%%%%%%%%%%%%%%%%%%%%%%%
\section{Main Results}\label{sec:main}

\subsection{Bandwidth Selection Consistency}

When there is a natural ordering in the data set, estimating the bandwidth of the precision matrix is important for detecting the dependence structure.
It is a crucial first step for the subsequent analysis.
In this subsection, we show the bandwidth selection consistency of the proposed prior.
Theorem \ref{theorem:band_sel} states that the posterior distribution puts a mass tending to one at the true bandwidth $k_0$.
Thus, we can detect the true bandwidth using the marginal posterior distribution for the bandwidth $k$.
We call this property the bandwidth selection consistency.

\begin{theorem}\label{theorem:band_sel}
	Consider model \eqref{model} and priors \eqref{prior1}--\eqref{prior3}.
	If conditions (A1)--(A4) are satisfied, then we have
	\bea
	\bbE_0 \Big[ \pi \big(k \neq k_0 \mid \bfX_n \big) \Big] &=& o(1) .
	\eea
\end{theorem}

Informed readers might be aware of the recent work of \cite{cao2017posterior} considering the selection of sparse Cholesky factors.
It should be noted that their method is not applicable to the bandwidth selection problem.
The key issue is that their method is not adaptive to the unknown sparsity corresponding to the true bandwidth $k_0$ in this paper: to obtain the selection consistency, the choice of hyperparameter should depend on $k_0$, which is unknown and of interest.
Furthermore, they required stronger conditions in terms of dimensionality $p$, true sparsity $k_0$, eigenvalues of the true precision matrix and beta-min for the strong model selection consistency.
%Even in the practical point of view, their method theoretically has the selection consistency for the sparse DAG models, which contain the banded DAG model as a special case, but the estimated DAG would not have the banded structure. 

%\cite{johnson2012bayesian} considered the strong model selection consistency for high-dimensional linear regression models.
%However, they considered only $p\le n$ case and assumed that the size of the true model is fixed.

\begin{remark}
	The bandwidth selection result does not necessarily imply the consistency of the Bayes factor.
	Note that prior \eqref{prior1}, $\pi(d_j) \propto d_j^{-1}$ and
	\bean
	\pi(k) &\propto& \prod_{j=2}^p \big(\what{d}_j^{(k)} \big)^{\frac{\tau n}{2} }, \label{pi_k_emp}
	\eean
	and priors \eqref{prior1}, \eqref{prior2} and $\pi(k)\propto 1$ lead to the same marginal posterior for $k$.
	Thus, the above priors also achieve the bandwidth selection consistency in Theorem \ref{theorem:band_sel}.
	%	The above empirical prior $\pi(k)$ decreases as $k$ gets larger, so it plays a role as a penalty term for large models.
	However, \eqref{pi_k_emp} might be inappropriate when the Bayes factor is of interest, because the  ratio of normalizing terms induced by prior \eqref{pi_k_emp} ($C_0$ and $C_1$ in \eqref{k0k1}) have a non-ignorable effect on the Bayes factor.
\end{remark}

%\subsection{Bayesian Tests for Bandwidth}

\subsection{Consistency of One-Sample Bandwidth Test}\label{subsec:tests}

In this subsection, we focus on constructing a Bayesian bandwidth test for the testing problem $H_0: k \le k^\star$ versus $H_1: k > k^\star$ for some given $k^\star$.
A Bayesian hypothesis test is based on the Bayes factor $B_{10}(\bfX_n)$ defined by the ratio of marginal likelihoods,
\bea
B_{10}(\bfX_n) &=& \frac{p(\bfX_n \mid H_1)}{p(\bfX_n \mid H_0)} .
\eea
We are interested in the consistency of the Bayes factor which is one of the most important asymptotic properties of the Bayes factor \citep{dass2004note}.
A Bayes factor is said to be {\it consistent} if $B_{10}(\bfX_n)$ converges to zero in probability under the true null hypothesis $H_0$ and $B_{10}(\bfX_n)^{-1}$ converges to zero in probability under the true alternative hypothesis $H_1$.
%In fact, the result presented in this paper is more powerful because we obtained the convergence rates of the Bayes factors under the both hypotheses.

Although the Bayes factor plays a crucial role in the Bayesian variable selection, its asymptotic behaviors in the high-dimensional setting are not well-understood \citep{moreno2010consistency}.
Few works studied the consistency of the Bayes factor in the high-dimensional settings \citep{moreno2010consistency,wang2014bayes,wang2016consistency}, which only focused on the pairwise consistency of the Bayes factor.
They considered the testing problem $H_0: k = k^{(0)}$ versus $H_1: k=k^{(1)}$ for any $k^{(0)}<k^{(1)}$, where $k$ is the number of nonzero elements of the linear regression coefficient.
Note that a Bayes factor is said to be {\it pairwise consistent}  if the Bayes factor $B_{10}(\bfX_n)$ is consistent for any pair of simple hypotheses $H_0$ and $H_1$.
%Bayesian test procedures in \cite{moreno2010consistency}, \cite{wang2014bayes} and \cite{wang2016consistency} might be applicable to the banded precision matrix problem via the modified Cholesky decomposition.
%However, we note here that the above pairwise comparison approaches are not appropriate for the bandwidth selection problem, because one has to impose prior distributions for every possible pair of bandwidths and check the induced Bayes factors.
%When the upper bound of bandwidth $R_n$ is not small, it can be computationally prohibitive.

We focus on the composite hypotheses $H_0: k \le k^\star$ and $H_1: k> k^\star$ rather than  simple hypotheses.
To conduct a Bayesian hypothesis test,  prior distributions for both hypotheses should be determined.
Denote the prior under the hypothesis $H_i$ as $\pi_i(A_n, D_n, k)$ for $i=0,1$.
Since the difference between two hypotheses comes only from the bandwidth, we will use the same conditional priors for $A_n$ and $D_n$ given $k$, i.e. $\pi_i(A_n, D_n, k) = \pi_i(k) \, \pi(A_n, D_n \mid k)$ for $i=0,1$, where $\pi(A_n, D_n \mid k)$ is chosen as \eqref{prior1} and \eqref{prior2}.
We suggest using priors $\pi_0(k)$ and $\pi_1(k)$ such that
\begin{equation}
\begin{split}\label{k0k1}
\pi_0(k) &= C_0^{-1} \pi(k)  ,\quad k=0,1,\ldots,k^\star  , \\
\pi_1(k) &= C_1^{-1} \pi(k)  ,\quad k=k^\star+1,\ldots,R_n, 
\end{split}
\end{equation}
where $C_0 = \sum_{k=0}^{k^\star}\pi(k)$ and $C_1 = \sum_{k=k^\star+1}^{R_n} \pi(k)$.
Then, the Bayes factor has the following analytic form,
\bea
B_{10}(\bfX_n) &=& \frac{ \sum_{k > k^\star} \int p(\bfX_n \mid \Omega_n, k) \pi (\Omega_n \mid k)\pi_1(k) d\Omega_n }{\sum_{k \le k^\star} \int p(\bfX_n \mid \Omega_n, k) \pi (\Omega_n\mid k) \pi_0(k) d\Omega_n} \\
&=& \frac{  \pi(k > k^\star\mid \bfX_n) }{ \pi(k \le k^\star\mid \bfX_n) } \times \frac{C_0 }{C_1 } ,
\eea
where the marginal posterior $\pi(k \mid \bfX_n)$ is given in \eqref{post} up to some normalizing constant.
Note that, the Bayes factor can be defined because both hypotheses have the same improper priors on $D_n$.
We will show that the Bayes factor is consistent for any composite hypotheses $H_0: k \le k^\star$ and $H_1: k> k^\star$, which is generally \emph{stronger than the pairwise consistency} of the Bayes factor.
If we assume that $\pi_1(k)/\pi_0(k') = O(1)$ for any $k$ and $k'$, then one can see that the consistency of the Bayes factor for hypotheses $H_0: k \le k^\star$ and $H_1: k> k^\star$ for any $k^\star$ implies the pairwise consistency of the Bayes factor for any pair of simple hypotheses $H_0:k =k^{(0)}$ and $H_1: k =k^{(1)}$ for $k^{(0)} < k^{(1)}$.

For given positive constants $M_{\rm bm}$, $\gamma$ and $\tau \in (0,0.4]$ and integers $R_n$, $k_0$ and $k^\star$, define
\bea
T_{n,H_0,k_0,k^\star} &=& k^\star \cdot \left\{ C_{\gamma,\tau}^{-1} \right\}^{ (k^\star+1-k_0)\cdot (p - \frac{R_n+k_0+1}{2} ) } , \\
T_{n,H_1,k_0,k^\star} &=& (R_n-k^\star) \cdot \left\{ C_{\gamma, M_{\rm bm}} \right\}^{ (k_0-k^\star)\cdot (p - \frac{k^\star+k_0+1}{2} ) }  ,
\eea
where $C_{\gamma,\tau}$ and $C_{\gamma, M_{\rm bm}}$ are defined in condition (A4).
Theorem \ref{theorem:one_band} shows the convergence rates of Bayes factors under each hypothesis.
It turns out that $\pi(k) = 1/(R_n+1)$ is sufficient for the consistency of the Bayes factor.

\begin{theorem}\label{theorem:one_band}
	Consider  model \eqref{model} and hypothesis testing problem $H_0: k \le k^\star$ versus $H_1: k > k^\star$.
	Assume priors \eqref{prior1} and \eqref{prior2} for $\pi(A_n, D_n \mid k)$ and the bandwidth priors in \eqref{k0k1} with $\pi(k)=1/(R_n+1)$. 
	If conditions (A1)--(A3) hold, $\tau > \gamma (1+\epsilon_1)/(1+\gamma)$ and $\exp(M_{\rm bm}) >2 \{(1+\gamma)/\gamma \}^{1/2}$,
	then the Bayes factor $B_{10}(\bfX_n)$ is consistent under $\bbP_0$.
	Moreover, under $H_0: k \le k^\star$, we have
	\bea
	B_{10}(\bfX_n)
	&=& O_p \big( T_{n,H_0,k_0,k^\star} \big), 
	\eea
	and under $H_1: k > k^\star$, 
	\bea
	B_{10}(\bfX_n)^{-1} &=& O_p(T_{n,H_1,k_0,k^\star}) . 
	\eea
\end{theorem}

\begin{remark}
	By choosing the prior $\pi(k) = 1/(R_n+1)$, it implies that $C_{\gamma,\tau}^{-1}$ and $C_{\gamma, M_{\rm bm}}$ are strictly smaller than $1$ by \eqref{pi_k_xi}.
	Since we assume that $p\to \infty$ as $n\to\infty$, one can easily check that $T_{n,H_0,k_0,k^\star}$ and $T_{n,H_1,k_0,k^\star}$ go to zero  as $n\to\infty$ under $H_0: k \le k^\star$ and $H_1: k > k^\star$, respectively.
\end{remark}

\begin{remark}
	Note that if we use prior \eqref{pi_k_xi} with $\xi\neq 1$, the effect of the prior, $C_0/C_1$, can dominate the posterior ratio, $\pi(k > k^\star\mid \bfX_n) / \pi(k \le k^\star\mid \bfX_n)$ in the Bayes factor.  
	Because the prior knowledge on the bandwidth is usually not sufficient, it is clearly undesirable.
	Moreover, the direction of effect is the opposite of the prior knowledge.
\end{remark}

\cite{an2014hypothesis} and \cite{cheng2017test} developed frequentist bandwidth tests for  the hypotheses $H_0: k \le k^\star$ versus $H_1: k > k^\star$ and showed that their test statistic is asymptotically normal under the null and has a power converging to one as $n\wedge p \to\infty$.
%They suggested backward procedures to detect the true bandwidth by controlling the familywise error rate.
Compared with the result in Theorem \ref{theorem:one_band}, \cite{cheng2017test} required the upper bound  $k_0 = O( [n/\log p]^{1/2} )$ for the true bandwidth $k_0$, which is much stronger than our condition (A3). 
\cite{an2014hypothesis} allowed $k_0 \le n-4$, but  assumed that the partial correlation coefficient between $X_{ij}$ and $X_{i,j-k_0}$ given $X_{i,j-k_0+1},\ldots,X_{i,j-1}$ is of order $o(n^{-1})$.
%$n \cdot\max_{k_0 <j \le p}\rho_j^2 (k_0)/(1- \max_{k_0 <j \le p}\rho_j^2 (k_0)) \to 0$ and $n\sqrt{p}\cdot \min_{k_0 <j \le p}\rho_j^2 (k_0)/(1- \min_{k_0 <j \le p}\rho_j^2 (k_0)) \to \infty$ as $n\to\infty$.
%If we assume the bounded eigenvalues for $\Omega_{0n}$,  $\rho_j(k_0)$ has the same rate with $a_{0, j j-k_0}$.
It implies that $\max_j |a_{0, j j-k_0}|$ converges to zero at some rate.
Thus, the nonzero elements $a_{0, j j-k_0}$, $j=k_0+1,\ldots,p$ {\it should} converge to zero, which is somewhat unnatural.

%We also note here that, due to the fixed significance level, the probability of type I error does not converge to zero even if the sample size $n$ goes to infinity.  
%somewhat unrealistic condition, $\max_j |a_{0, j j-k_0}|=o(n^{-1/2})$, 
%while conditions in this paper allow $\max_j |a_{0, j j-k_0}|=O(1)$.

%\cite{cheng2017test} also considered the same hypothesis testing problem and proposed a consistent test without assuming parametric model.
%From the page 62 in \cite{anderson2003introduction}, the $p$-dimensional random vector $X = (X_1,\ldots,X_p)^T$ with mean zero and precision matrix $\Omega$ can be written as 
%\bea
%X_i &=& X_{-i}^T \beta_i + \epsilon_i , \quad i=1,\ldots,p ,
%\eea
%where $X_{-i}=(X_1,\ldots,X_{i-1},X_{i+1},\ldots,X_p)^T \in \bbR^{p-1}$, $\beta_i= - \Omega_{-i, i}/\Omega_{ii}\in \bbR^{p-1}$  and $\Omega_{-i, i} = (\Omega_{1,i},\ldots,\Omega_{i-1,i}, \Omega_{i+1,i},\ldots,\Omega_{p,i})^T$.
%Here, $\epsilon_i$ is the error uncorrelated of $X_{-i}$ satisfying $\C(\epsilon_i, \epsilon_j)=\Omega_{ij}/(\Omega_{ii}\Omega_{jj} )$.
%Based on the above interpretation, the bandwidth of $\Omega$ is equal to that of $\V(\epsilon)$, where $\epsilon=(\epsilon_1,\ldots,\epsilon_p)^T$.
%They constructed a test statistic depending on a consistent estimator of the precision matrix  and showed theoretical properties of the test.
%To detect the true bandwidth, similar to \cite{an2014hypothesis}, some correction methods for multiple testing can be applied to their test.
%and assumed additional conditions on $\Omega_{0n}$ and $\V(\epsilon) \in \bbR^{p\times p}$.

\cite{johnson2010use,johnson2012bayesian} and \cite{rossell2017tractable} pointed out that the use of local alternative prior leads to imbalanced convergence rates for the Bayes factors, and showed that this issue can be avoided by using  non-local alternative priors. 
However, interestingly, convergence rates for the Bayes factors in Theorem \ref{theorem:one_band} yield \emph{similar order of rates} under  both hypotheses without using a non-local prior.
Roughly speaking, 
%at least in the model and prior considered in this paper, 
the imbalance issue can be ameliorated  by introducing the beta-min condition (Condition (A2)).
To simplify the situation, consider the model
\bea
Y &=& X \beta^{(k)} + \epsilon,
\eea
where $Y = (Y_1,\ldots,Y_n)^T$, $X\in \bbR^{n\times p}$, $\beta^{(k)} = (\beta_1,\ldots,\beta_k ,0,\ldots,0)^T \in \bbR^p$, $\epsilon=(\epsilon_1,\ldots,\epsilon_n)^T$ and $\epsilon_i \overset{i.i.d.}{\sim} N(0, \sigma^2)$. 
Suppose  priors  \eqref{prior1} and \eqref{prior2} are imposed on $(\beta_1,\ldots,\beta_k)^T$ and $\sigma^2$ given $k$. 
Consider  hypotheses $H_0: k=k_1$ and $H_1: k=k_2$, where $k_1 < k_2$, and assume that the eigenvalues of $X_{(1:k_2)}$ are bounded and $k_2-k_1 \to \infty$ as $n\to\infty$ for simplicity.
Note that the prior for $\beta^{(k_2)}$ is a local alternative prior  because  $\pi(\beta^{(k_2)} \mid \sigma^2)>c$ on $\{ \beta^{(k_2)} \in \bbR^{k_2} : \beta^{(k_2)} =(\beta_1,\ldots,\beta_{k_1},0,\ldots,0)^T \}$ for some constant $c>0$.
If $H_0$ is true, $B_{10}(Y)$ decreases at rate $O_p( e^{-c_0(k_2-k_1)} )$ for some constant $c_0>0$ based on techniques in the proof of Theorem \ref{theorem:band_sel}.
%regardless of the existence of beta-min condition. 
On the other hand, if $H_1$ is true, $B_{10}(Y)^{-1}$ decreases exponentially with $n (k_2-k_1) \beta_{\rm min}^2$, where $\beta_{\rm min}$ is the lower bound for the absolute of nonzero elements of $\beta_0^{(k_2)}$.
\cite{johnson2010use,johnson2012bayesian} and \cite{rossell2017tractable}  assumed that $\beta_{\rm min}^2 > c_1$ for some constant $c_1>0$.
In that case, $B_{10}(Y)^{-1}$ decreases exponentially with $n (k_2-k_1) c_1$, which causes the imbalanced convergence rates.
However, if we assume  $\beta_{\rm min}^2 \ge c_2 n^{-1}$ similar to condition (A2), $B_{10}(Y)^{-1}$ decreases at rate $O_p( e^{-c_2(k_2-k_1)} )$ for some constant $c_2>0$.
Thus, convergence rates for the Bayes factors have similar order under the both hypotheses.

The above argument does not mean that the non-local priors are not useful for our problem. 
We note that the balanced convergence rates by using the beta-min condition is different from those by using the non-local prior.
The former makes the rate of $B_{10}(Y)^{-1}$ slower under $H_1$, while the latter makes the rate of $B_{10}(Y)$ faster under $H_0$.
Thus, the use of non-local priors might improve the rates of convergence for $B_{10}(Y)$ under $H_0$ in Theorem \ref{theorem:one_band}.
However, it will increase the computational burden and is unclear which rate one can achieve using the non-local prior under condition (A2), so we leave it as a future work.

\subsection{Consistency of Two-Sample Bandwidth Test}\label{subsec:two_tests}

Suppose we have two data sets from the models
\bean\label{model_two}
\begin{split}
	X_1,\ldots, X_{n_1} \mid \Omega_{1n_1} &\overset{i.i.d.}{\sim} N_p(0, \Omega_{1n_1}^{-1}) , \\
	Y_1,\ldots, Y_{n_2} \mid \Omega_{2 n_2} &\overset{i.i.d.}{\sim} N_p(0, \Omega_{2n_2}^{-1}),
\end{split}
\eean
where $\Omega_{1n_1} = (I_p- A_{1n_1})^T D_{1n_1}^{-1}(I_p - A_{1n_1})$ and $\Omega_{2n_2} = (I_p- A_{2n_2})^T D_{2n_2}^{-1}(I_p - A_{2n_2})$ are the MCDs. 
Denote the bandwidth of $\Omega_{i n_i}$ as $k_i$ for $i=1,2$.
In this subsection, our interest is the test of equality between two bandwidths $k_1$ and $k_2$, the {\it two-sample bandwidth test}.
We consider the hypothesis testing problem $H_0: k_1 = k_2$ versus $H_1: k_1 \neq k_2$ and investigate the asymptotic behavior of the Bayes factor,
\bea
B_{10}(\bfX_{n_1}, \bfY_{n_2}) &=& \frac{p(\bfX_{n_1}, \bfY_{n_2} \mid H_1)}{p(\bfX_{n_1}, \bfY_{n_2} \mid H_0)} ,
\eea
where $\bfX_{n_1} = (X_1^T, \ldots, X_{n_1}^T)^T \in \bbR^{n_1 \times p}$ and $\bfY_{n_2} = (Y_1^T, \ldots, Y_{n_2}^T)^T \in \bbR^{n_2 \times p}$.
Suppose that multivariate observations are collected from two populations, and a test of the equality of dependence structure is the main interest.
When the dependence structure is directly related to how many previous variables influencing the current variable,  two-sample bandwidth test provides a suitable answer.

Denote the priors under $H_0$ and $H_1$ as, respectively
\bea
&& \pi_0 \big(A_{1 n_1}, D_{1 n_1}, A_{2 n_2}, D_{2 n_2} \mid k \big)  \\
&=& \pi \big(A_{1 n_1}, D_{1  n_1} \mid k \big) \, \pi \big(A_{2 n_2}, D_{2 n_2}\mid k \big) \, \pi_0(k), \quad k=0,1,\ldots,R_n,
\eea
and 
\bea
&& \pi_1 \big(A_{1 n_1}, D_{1 n_1}, A_{2 n_2}, D_{2 n_2} \mid k_1, k_2 \big)  \\
&=& \pi\big(A_{1 n_1}, D_{1  n_1} \mid k_1\big) \, \pi \big(A_{2 n_2}, D_{2 n_2}\mid k_2\big) \, \pi_1(k_1, k_2) , \quad 0\le k_1 \neq k_2 \le R_n.
\eea

We suggest the following conditional priors  $\pi(A_{1 n_1}, D_{1  n_1} \mid k_1)$ and $\pi(A_{2 n_2}, D_{2 n_2}\mid k_2)$ for any given $k_1$ and $k_2$,
\bean\label{twesam_prior}
\begin{split}
	a_{1,j}^{(k_1)} \mid d_{1,j} , k_1 \,\,&\overset{ind}{\sim}\,\, N_{k_{1j}} \Big(\what{a}_{1,j}^{(k_1)} , \frac{d_{1,j}}{\gamma} \big( \bfX_{j (k_1)}^T \bfX_{j (k_1)} \big)^{-1}  \Big) , \\
	\pi(d_{1,j}) \,\,&\overset{i.i.d.}{\propto}\,\, d_{1,j}^{\tau n_1/2 - 1} , \\
	a_{2,j}^{(k_2)} \mid d_{2,j} , k_2 \,\,&\overset{ind}{\sim}\,\, N_{k_{2j}} \Big(\what{a}_{2,j}^{(k_2)} , \frac{d_{2,j}}{\gamma} \big( \bfY_{j (k_2)}^T \bfY_{j (k_2)} \big)^{-1}  \Big) , \\
	\pi(d_{2,j}) \,\,&\overset{i.i.d.}{\propto}\,\, d_{2,j}^{\tau n_2/2 - 1} ,
\end{split}
\eean
where $k_{ij}= k_i \wedge (j-1)$, $a_{i,j}^{(k_i)} \in \bbR^{k_{ij}}$ is the nonzero elements in the $j$th row of $A_{i n_i}$ and $D_{i n_i} = diag(d_{i,j})$ for $i=1,2$.
Similar to the previous notations, we denote $\what{a}_{1,j}^{(k_1)} = (\bfX_{j (k_1)}^T \bfX_{j (k_1)} )^{-1} \bfX_{j (k_1)}^T \tilde{X}_j$ and $\what{a}_{2,j}^{(k_2)} = (\bfY_{j (k_2)}^T \bfY_{j (k_2)} )^{-1} \bfY_{j (k_2)}^T \tilde{Y}_j$, where $\bfY_{j (k_2)}\in \bbR^{n\times k_2}$ is the sub-matrix consisting of $(j-k_2)_1,\ldots, (j-1)$th columns of $\bfY_n$.
The priors on bandwidths are chosen as
\begin{equation}
\begin{split}\label{twosample_pik}
\pi_0(k) &=  \frac{1}{R_n+1}, \quad k=0,1,\ldots,R_n, \\
\pi_1(k_1, k_2) &= \frac{1}{R_n(R_n+1)}, \quad 0 \le k_1 \neq k_2 \le R_n .
\end{split}
\end{equation}
This choice of priors leads to an analytic form of the Bayes factor,
\bea
B_{10}(\bfX_{n_1}, \bfY_{n_2}) &=& \frac{\sum_{k_1\neq k_2} \pi(k_1\mid \bfX_{n_1})\pi(k_2\mid \bfY_{n_2})  }{\sum_{k_1= k_2} \pi(k_1\mid \bfX_{n_1})\pi(k_2\mid \bfY_{n_2})} \times R_n^{-1} ,
\eea
where the marginal posterior distributions $\pi(k_1\mid \bfX_{n_1})$ and $\pi(k_2 \mid \bfY_{n_2})$ are known up to some normalizing constants similar to \eqref{post}.
We denote $\Omega_{0,i n_i}$ as the true precision matrix with bandwidth $k_{0i}$ for $i=1,2$ and assume that $p$ tends to infinity as  $n = n_1 \wedge n_2 \to \infty$.
Theorem \ref{theorem:two_band} gives a sufficient condition for the consistency of the Bayes factor $B_{10}(\bfX_{n_1}, \bfY_{n_2})$ by calculating the convergence rates.

\begin{theorem}\label{theorem:two_band}
	Consider  model \eqref{model_two} and hypotheses $H_0: k_1=k_2$ and $H_1:k_1 \neq k_2$.
	Assume the conditional priors given bandwidths \eqref{twesam_prior} and the bandwidth priors \eqref{twosample_pik}. 
	If conditions (A1)--(A3) for $\Omega_{0,1n_1}$, $\Omega_{0,2n_2}$ and priors are satisfied, $\tau > \gamma (1+\epsilon_1)/(1+\gamma)$ and $\exp(M_{\rm bm}) >2 \{(1+\gamma)/\gamma \}^{1/2}$, then the Bayes factor $B_{10}(\bfX_{n_1}, \bfY_{n_2})$ is consistent under $\bbP_0$.
	Moreover, under $H_0:k_1=k_2$, we have
	\bea
	B_{10}(\bfX_{n_1}, \bfY_{n_2})
	&=&  O_p \left( \frac{k_0}{R_n-k_{0}} T_{n,H_1,k_0, k_0-1} + \frac{R_n-k_0}{k_0} T_{n,H_0, k_0, k_0}   \right),
	\eea
	and under $H_1: k_1\neq k_2$, 
	\bea
	B_{10}(\bfX_{n_1}, \bfY_{n_2})^{-1}
	&=& O_p \left(  \frac{R_n \, k_{\min}}{R_n-k_{\min}} T_{n,H_1, k_{\min}, k_{\min}-1} + \frac{R_n(R_n-k_{\min})}{k_{\min}} T_{n,H_0, k_{\min},k_{\min}} \right),
	\eea
	where $k_{\min} = k_{01}\wedge k_{02}$.
\end{theorem}

As mentioned earlier, to the best of our knowledge, this is the first consistent two-sample bandwidth test result in high-dimensional settings.
Frequentist testing procedures in \cite{an2014hypothesis} and \cite{cheng2017test} focused only on the one-sample bandwidth test, and it is unclear whether these methods can be extended to the two-sample testing problem.
%Correction procedures for the multiple testing problem such as familywise error rate and false discovery rate controlling procedures might be applied for conducting two sample bandwidth test, 
%but  at the some cost depending on the type of procedures.

Note that the hypothesis testing problem $H_0:k_1=k_2$ versus $H_1:k_1 \neq k_2$ is different from the hypothesis testing $H_0: \Omega_{1n_1} = \Omega_{2 n_2}$ versus $H_1: \Omega_{1n_1} \neq \Omega_{2 n_2}$ in \cite{cai2013two}.
The latter testing problem is called the two-sample precision (or covariance) test.
The two-sample bandwidth test is weaker than the two-sample precision test, i.e. if the two-sample bandwidth test supports the null hypothesis, then one can further conduct the two-sample precision test.

\section{Numerical Results}\label{sec:numerical}

We have proved the bandwidth selection consistency and  convergence rates of Bayes factors based on  priors \eqref{prior1}--\eqref{prior3}. 
In this section, we conduct simulation studies to describe the practical performance of the proposed method.
Throughout the section, we use the prior $\pi(k) = 1/(R_n+1)$.

\subsection{Comparison with other Bandwidth Tests}\label{subsec:comparison}

In this subsection, we compared the performance of our method with those of other bandwidth selection procedures.
Since we have bandwidth selection consistency (Theorem \ref{theorem:band_sel}), we suggest using the posterior mode to estimate the true bandwidth $k_0$.
We chose the bandwidth test of \cite{an2014hypothesis} as a frequentist competitor and the bandwidth selection procedures of \cite{banerjee2014posterior} and \cite{lee2017estimating} as Bayesian competitors.
Significance levels for bandwidth tests in \cite{an2014hypothesis} were varied $\alpha=0.001, 0.005, 0.01$, but only the result with $\alpha=0.01$ are reported since they gave similar results.
%\cite{an2014hypothesis} proposed bandwidth selection procedures, algorithms 1 and 2, which we denote BA1 and BA2, respectively, in Table \ref{table:comp1}.
For \cite{banerjee2014posterior} and \cite{lee2017estimating}, we used the prior $\pi(k) \propto \exp(-k^4)$ as they suggested.
Note that these Bayesian procedures do not guarantee the bandwidth selection consistency.

To calculate the marginal posterior in \eqref{post}, the hyperparameters $\gamma, \tau$ and $R_n$ should be determined.
As a pragmatic approach, we incorporated cross-validation (CV) to select $\gamma$, and fixed $\tau=0.01$ and $R_n=k_0+10$; in our experiments, we also tried $\tau=0.01,0.02,\ldots,0.30$ and selected $\tau$ via CV, but found that $\tau=0.01$ is selected in most cases.
We randomly divided the data $\bfX_n$ into two subsamples, the test set $\bfX_{n_1}^{te}$ and training set $\bfX_{n_2}^{tr}$, where $n_1 = \lceil n/3 \rceil$ and $n_2 = n - n_1$.
Let $\what{k}(\gamma)$ be the posterior mode based on $\bfX_{n_2}^{tr}$ and a given $\gamma$, and define the mean squared error (MSE)
\bea
MSE(\gamma) &=& \sum_{j=1}^p \| \tilde{X}_j^{te} - \bfX_{j (\what{k}(\gamma))}^{te} \what{a}_j^{(\what{k}(\gamma))} \|_2^2.
\eea
We divided the data 20 times and selected $\what{\gamma}$ which minimizes $20^{-1} \sum_{\nu=1}^{20}MSE_\nu(\gamma)$, where $MSE_\nu(\gamma)$ is the MSE from the $\nu$th subsampling. 
Depending on the purpose of the analysis, other criteria besides MSE can be adopted.

The data sets were generated from $N_p(0, \Omega_{0n}^{-1})$, where $\Omega_{0n} = (I_p - A_{0n})^T D_{0n}^{-1} (I_p - A_{0n})$.
For each $j=k_0+1,\ldots,p$,  nonzero elements in the $j$th row of the true Cholesky factor $A_{0n}$ were sampled from $Unif(A_{0,\min}, A_{0,\max})$ and ordered to satisfy $a_{0,jl} \le a_{0,jl'}$ for any $l < l'$.
The diagonal elements of $D_{0n}$ were generated from $Unif(5, 10)$.
To investigate performance in various settings, the values of $n,p, k_0, A_{0,\min}$ and $A_{0, \max}$ were varied.
The simulation results, based on 50 simulated data sets for each setting, are reported in Table \ref{table:comp1} and Figure \ref{fig:comparison}.
We denoted the proposed method in this paper as BBS, the Bayesian Bandwidth Selector.

\begin{table}[!tb]
	\centering\footnotesize
	\caption{
		The summary statistics for each setting are represented, where $k_0=5$ and $[A_{0,\min},A_{0,\max}]=[0.1, 0.1]$. 
		BBS: the proposed method in this paper. 
		LL: bandwidth selection procedure of \cite{lee2017estimating}. 
		BG: bandwidth selection procedure of \cite{banerjee2014posterior}.
		BA1 and BA2: algorithms 1 and 2 in \cite{an2014hypothesis}, respectively.
	}\vspace{.15cm}
	\begin{tabular}{c c c c c}
		\hline 
		& \multicolumn{4}{c}{$\big(\, \what{p}_0 , \,\, \what{k}_0 \,\big)$} \\ 
		& $(n=70, p=100)$ & $(n=70, p=200)$ & $(n=200, p=100)$ & $(n=200, p=200)$ \\ \hline
		BBS & $(0.96 ,\, 4.94)$ & $(0.98, \, 4.98)$  & $(1.00,\, 5.00)$    & $(1.00,\, 5.00)$  \\ \hline
		BA1  & $(0.78 ,\, 4.84)$ & $(0.96,\, 5.16)$  &   $(0.96,\, 5.06)$  & $(1.00,\, 5.00)$ \\ \hline
		BA2  & $(0.78 ,\, 4.84)$ & $(0.98,\, 5.18)$  &  $(0.96,\, 5.06)$   & $(1.00,\, 5.00)$  \\ \hline
		LL & $(0.00 ,\,  1.00)$ & $(0.00 ,\,  1.00)$   &  $(0.00,\, 1.02)$  & $(0.00, 1.04)$  \\ \hline
		BG & $(0.00 ,\,  1.00)$ & $(0.00 ,\,  1.00)$   &  $(0.00,\,1.00)$  &  $(0.00, 1.00)$   \\ \hline
	\end{tabular}\label{table:comp1}
\end{table}
\begin{table}[!tb]
	\centering\footnotesize
	\caption{
		The summary statistics for each setting are represented, where $k_0=10$ and $[A_{0,\min},A_{0,\max}]=[0.1, 0.2]$. 
	}\vspace{.15cm}
	\begin{tabular}{c c c c c}
		\hline 
		& \multicolumn{4}{c}{$\big(\, \what{p}_0 , \,\, \what{k}_0 \,\big)$} \\ 
		& $(n=70, p=100)$ & $(n=70, p=200)$ & $(n=200, p=100)$ & $(n=200, p=200)$ \\ \hline
		BBS & $(0.94 ,\, 10.02)$ & $(0.96, \, 10.06)$  & $(1.00,\, 10.00)$    & $(1.00,\, 10.00)$  \\ \hline
		BA1  & $(0.72 ,\, 9.66)$ & $(0.98,\, 10.14)$  &   $(1.00,\, 10.00)$  & $(0.98,\, 10.06)$ \\ \hline
		BA2  & $(0.74 ,\, 9.72)$ & $(0.98,\, 10.14)$  &  $(1.00,\, 10.00)$   & $(0.96,\, 10.08)$  \\ \hline
		LL & $(0.00 ,\,  2.12)$ & $(0.00 ,\,  2.86)$   &  $(0.00,\, 3.36)$  & $(0.00, 4.00)$  \\ \hline
		BG & $(0.00 ,\,  1.00)$ & $(0.00 ,\,  1.00)$   &  $(0.00,\,1.00)$  &  $(0.00, 1.00)$   \\ \hline
	\end{tabular}\label{table:comp2}
\end{table}
\begin{figure*}[!tb]
	\centering
	\includegraphics[width=7.5cm,height=6.cm]{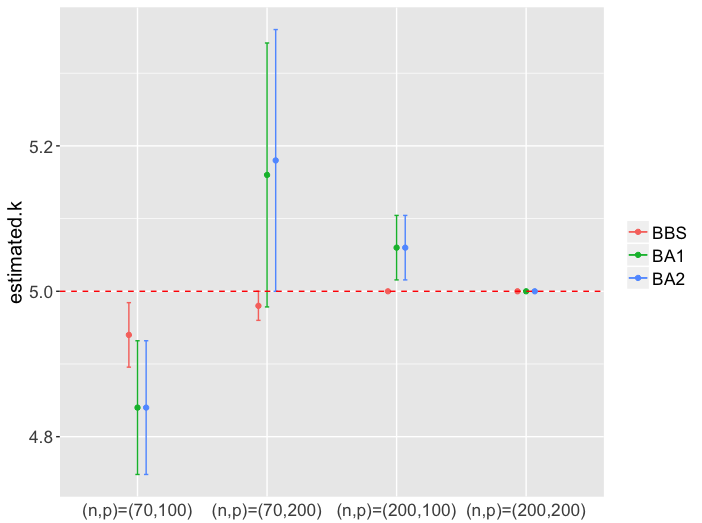}
	\includegraphics[width=7.5cm,height=6.cm]{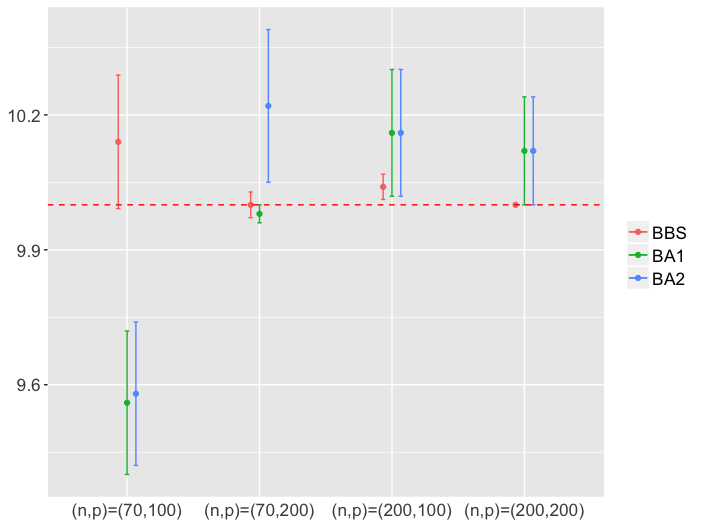}
	\includegraphics[width=7.5cm,height=6.cm]{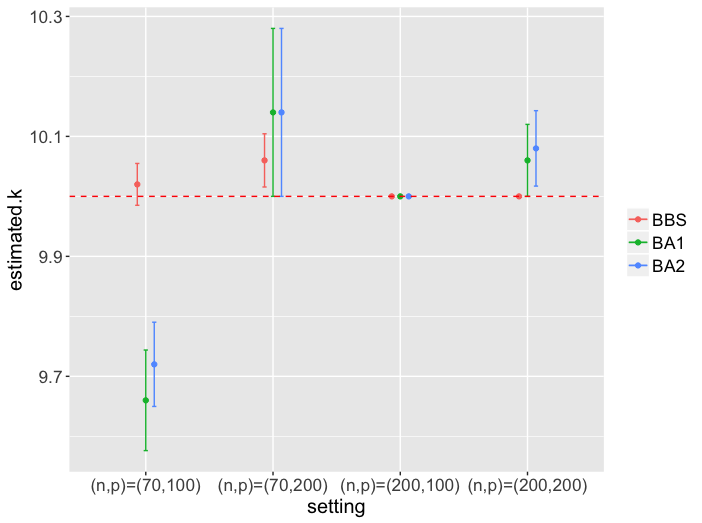}
	\includegraphics[width=7.5cm,height=6.cm]{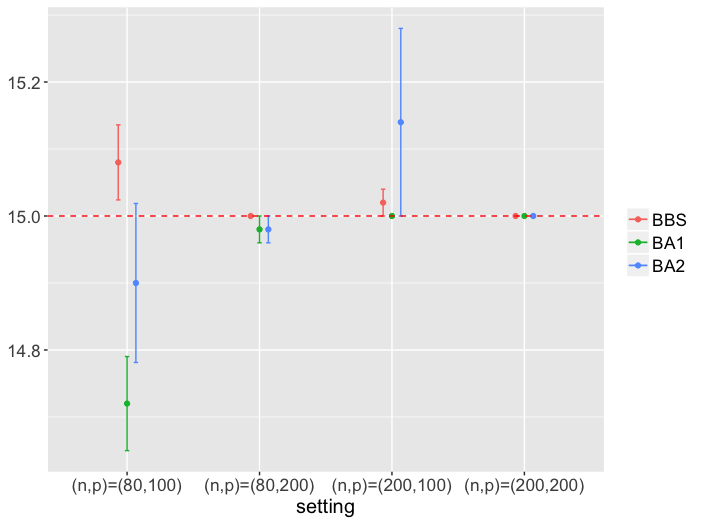}
%	\vspace{-.5cm}
	\caption{The summary plots for estimated bandwidth in various settings. 
	The middle dot and bar represent the mean and standard error of the mean, respectively, based on 50 simulations.
	$[A_{0,\min},A_{0,\max}]=[0.1, 0.1]$ was used for the top row, while $[A_{0,\min},A_{0,\max}]=[0.1, 0.2]$ was used for the bottom row.
	The red dashed line is the true bandwidth.}
	\label{fig:comparison}
\end{figure*}

Performance of each method were evaluated by the proportion of correct detections of $k_0$, $\what{p}_0 = \sum_{s=1}^{50}I(\hat{k}_0^{(s)} = k_0)/50 $, and averaged bandwidth estimate, $\what{k}_0 = \sum_{s=1}^{50}\hat{k}_0^{(s)} /50$, where $\hat{k}_0^{(s)}$ is the estimated bandwidth for the $s$th data set.
Our method, the BBS, consistently outperformed other competitors in most settings.
The bandwidth selection procedures of \cite{an2014hypothesis} worked reasonably well for large $n$ and large $p$ cases, but it seems somewhat unstable when $(n=70, p=100)$.
Although \cite{lee2017estimating} is slightly better than \cite{banerjee2014posterior},  both of them consistently underestimated the true bandwidth $k_0$.
The proposed prior $\pi(k) \propto \exp(-k^4)$ seems to be too strong to put sufficient masses near the true bandwidth $k_0$ especially when $k_0$ is not small.
Figure \ref{fig:comparison} shows the bandwidth selection results of BBS and the test in \cite{an2014hypothesis} to compare the performance of the two methods at a glance.
As shown in Tables \ref{table:comp1} and \ref{table:comp2}, the BBS outperformed the bandwidth tests in \cite{an2014hypothesis} in most cases.

\subsection{Telephone Call Center Data}\label{subsec:realdata}

We illustrate the performance of the proposed method using the telephone call center data previously analyzed by \cite{huang2006covariance}, \cite{bickel2008regularized} and \cite{an2014hypothesis}.
The phone calls were recorded from 7:00 am until midnight from a call center of a major U.S. financial organization.
The data were collected for  239 days in 2002 except holidays, weekends and days when the recording system did not work properly.
The number of calls were counted for every 10 minutes, and a total of 102 intervals were obtained on each day.
We denote the number of calls on the $j$th time interval of the $i$th day as $N_{ij}$ for each $i=1,\ldots, 239$ and $j=1,\ldots,102$.
As in \cite{huang2006covariance}, \cite{bickel2008regularized} and \cite{an2014hypothesis}, a transformation $X_{ij} = \sqrt{N_{ij} + 1/4}$ was applied to make the data close to the random sample from normal distribution.
The transformed data set was centered.
For more details about the data set, see \cite{huang2006covariance}.

We are interested in predicting the number of phone calls from the $52$nd to $102$nd time intervals using the previous counts on each day.
The best linear predictor of $X_{ij}$ from $X_{i}^j = (X_{i1},\ldots, X_{i,j-1})^T$,
\bean\label{bestlinear}
\what{X}_{ij} &=& \mu_j + \sg_{(j, 1:(j-1))} \big[\sg_{(1:(j-1), 1:(j-1))} \big]^{-1} (X_{i}^j  - \mu^j )    ,
\eean
was used to predict $X_{ij}$ for each $j=52,\ldots, 102$, where $\mu_j = \bbE(X_{1j})$, $\mu^j = (\mu_1,\ldots, \mu_{j-1})^T$ and $\sg_{S_1, S_2}$ is a sub-matrix of $\sg$ consisting of the $S_1$th rows and the $S_2$th columns for given index sets $S_1$ and $S_2$.
We used the first 205 days $(i=1,\ldots,205)$ as a training set and the last 34 days $(i=206,\ldots,239)$ as a test set.
To calculate the best linear predictor  \eqref{bestlinear}, the unknown parameters are need to be estimated.
Because it is reasonable to assume the existence of the natural (time) ordering, we plugged the estimators $\what{\mu}^j =  \sum_{i=1}^{205} X_{i}^j /205$ and $\what{\sg}_k = \big\{ (I_p - \what{A}_{nk} )^T \what{D}_{nk}^{-1} (I_p - \what{A}_{nk} ) \big\}^{-1}$ into \eqref{bestlinear}, where $\what{A}_{nk}$ and $\what{D}_{nk}$ are estimators based on the training set.

We applied the proposed methods in this paper, \cite{an2014hypothesis} and \cite{bickel2008regularized} to estimate the bandwidth $k$ using the training set, and compared the prediction errors $\text{PE}_j = \sum_{i=206}^{239} |\what{X}_{ij} - X_{ij}|/34$ for each $j=52,\ldots,102$.
We defined the average of prediction errors, $\sum_{j=52}^{102} \text{PE}_j/51$ to illustrate the performance of estimated bandwidths.
For a fair comparison, we used the same estimator $\what{\sg}_k$ and only chose different bandwidths depending on the selection procedure.
Since the goal of the analysis is prediction, the average of  prediction errors using training set was used as the criterion for CV.
Based on the selected hyperparameter, our method, the BBS, gives the estimated bandwidth $\what{k} = 5$.
Algorithms 1 and 2 with $\alpha=0.01$ in \cite{an2014hypothesis} determined the bandwidth as 8 and 10, respectively, and \cite{bickel2008regularized} selected the bandwidth as 19 based on a resampling scheme proposed in their paper.
The average of prediction errors were $0.5347, 0.5474, 0.5568$ and $0.5609$ at bandwidth $k=5$, $8$, $10$ and $19$, respectively.
Note that if we use the sample covariance matrix instead of the banded estimator $\what{\sg}_k$, it gives the average prediction error 0.7008.
Thus, the banded estimator of $\sg$ benefits in this case, and our bandwidth estimate yields smaller average prediction error compared with other procedures.
\begin{figure*}[!tb]
	\centering
	\includegraphics[width=11cm,height=6cm]{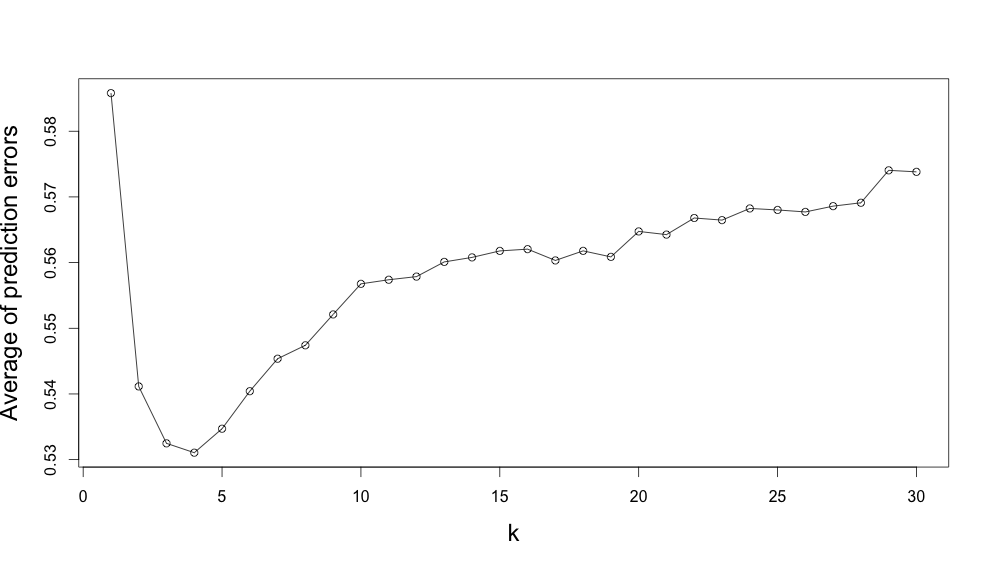}
	\vspace{-.5cm}
	\caption{The averages of prediction errors are represented for various bandwidth values $k$.}
	\label{fig:EPplot}
\end{figure*} 
Figure \ref{fig:EPplot} represents the averages of prediction errors for various bandwidth values $k$.
The minimum error is attained at $k=4$.
None of the above methods achieves the optimal bandwidth $k=4$, but the bandwidth obtained from our method is closest to 4.

\section{Discussion}\label{sec:discussion}

%In this paper we introduced a prior distribution for high-dimensional banded precision matrices with  primary interests on  Bayesian bandwidth selection and tests for one- and two-sample problems. 
%The induced posterior distribution attains the strong model selection consistency under mild conditions.
%We also proved the consistency of Bayes factors for one- and two-sample bandwidth tests.
%The proposed bandwidth selection procedure outperforms other Bayesian procedures of \cite{banerjee2014posterior} and \cite{lee2017estimating}, and comparable to frequentist test of \cite{an2014hypothesis}.

Throughout the paper, we assumed that each row of the Cholesky factor has the same bandwidth for simplicity.
It can be extended to more general setting allowing different bandwidth for each row.
If we denote the bandwidth for the $j$th row as $k_{0}^{(j)}$ and $k_{0,\max} = \max_{1\le j \le p} k_{0}^{(j)}$, then one can conduct the bandwidth test for $k_{0,\max}$.
Theoretical results in this paper also hold for the maximum bandwidth $k_{0,\max}$ selection problem  with possibly some additional conditions.
For example, if $k_0^{(j)} = k_{0,\max}$ except only finite $j$'s, then the proposed priors still achieve the theoretical properties in Section \ref{sec:main}.
%It should also be noted that the proposed prior can be adopted for 

The bandwidth selection problem for \emph{bandable matrices} is one of the interesting future research topics.
Note that it has very different characteristics from that for banded matrices. 
In the bandable case, the bandwidth selection is to find the optimal bandwidth minimizing the estimation error with respect to some loss function.
It is well known that the optimal bandwidth depends on the loss function \citep{cai2010optimal}.
Thus, if the bandwidth selection of the bandable matrix is of primary interest, the prior distribution should be chosen carefully depending on the loss function.

\section*{Acknowledgement}

We thank Baiguo An for providing us the telephone call center data.
%We gratefully acknowledge the funding support from NSF grants IIS 1663870,  DMS CAREER 1654579 and an ARO award W911NF-15-1-0440.

%\section*{Supplementary material}
%\label{SM}
%Supplementary material available at \Bka\ online includes the proofs
%for main results and other auxiliary results.

%\section*{Appendix}

\appendix

%\appendixone

\section*{Appendix 1: Posterior convergence rate for the Cholesky factor}\label{sec:est_chol}

The estimation of Cholesky factor is important to detect the dependence structure of data. Although our primary goals are Bayesian bandwidth test and model selection,  we show in this section that the proposed prior can be used to estimate Cholesky factors.
Theorem \ref{theorem:band_sel} implies that $\what{k} = \argmax_{0\le k\le R_n} \pi(k \mid \bfX_n)$ is a consistent estimator of $k_0$.
Consider an empirical Bayes approach by considering priors \eqref{prior1} and \eqref{prior2} with $\what{k}$ instead of imposing a prior on $k$.
This empirical Bayes method faciliates easy implementations when the estimation of Cholesky factor or precision matrix is of interest.
To assess the performance, we adopt the {\it P-loss convergence rate} used by \cite{castillo2014bayesian} and \cite{lee2018optimal}.
Corollary \ref{cor:Ploss_rate} presents the P-loss convergence rate of the empirical Bayes approach with respect to the Cholesky factor under the matrix $\ell_\infty$-norm.
We denote $\pi_{(k)}$ as the empirical prior stated above and $\bbE^{\pi_{(k)}}(\cdot \mid \bfX_n)$ as the posterior expectation induced the prior $\pi_{(k)}$.

\begin{corollary}\label{cor:Ploss_rate}
	Consider  model \eqref{model} and priors \eqref{prior1} and \eqref{prior2} with $\what{k}$ instead of $k$.
	If conditions (A1)--(A4) are satisfied and $k_0 + \log p = o(n)$, then we have
	\bea
	\bbE_0 \bbE^{\pi_{(k)}} \Big( \| A_{n} - A_{0n} \|_\infty  \mid \bfX_n \Big) &\lesssim& \left(\frac{k_0( k_0 + \log p)}{n} \right)^{\frac{1}{2}} .
	\eea
\end{corollary}

Define a class of precision matrices
\bea
{\cal{U}}_p \,\,=\,\, {\cal{U}}_p(\epsilon_{0n}, \zeta_{0n}, M_{\rm bm}, \tau, R_n)  &=& \bigg\{ \Omega \in \calC_p: \,\, \Omega \text{ satisfies (A1)}--(A3)\,\,  \bigg\} ,
\eea
where $\calC_p$ is the class of $p\times p$ symmetric positive definite matrices.
%With a slight modification of Example 13.12 in an unpublished lecture note of John Duchi \citep{duchi2016lecture}, the minimax lower bound is given by
With a slight modification of Example 13.12 in an unpublished lecture note of John Duchi, the minimax lower bound is given by
\bea
\inf_k \inf_{\what{A}_{nk}} \sup_{\Omega_{0n} \in {\cal{U}}_p} \bbE_0  \Big( \| \what{A}_{nk} - A_{0n} \|_\infty \Big) &\gtrsim& \frac{k_0}{\sqrt{n}},
\eea
where the second infimum is taken over all estimators with bandwidth $k$.
Thus, the above empirical Bayes approach achieves nearly optimal P-loss convergence rate.
%We conjecture that the additional $\log p$ term is unavoidable because it is essential to ensure the sparse Riesz condition for $\bfX_n$ defined at Lemma 1 in \cite{lee2017estimating}.

\section*{Appendix 2: Proofs}\label{sec:proofs}

%\subsection{Proof for Bandwidth Selection Consistency}\label{subsec:proof_sel}

\begin{proof}[Proof of Theorem \ref{theorem:band_sel}]
	For any $k=0,1,\ldots,R_n$, we have
	\bea
	\pi(k\mid \bfX_n) &\le& \frac{\pi(k\mid \bfX_n)}{\pi(k_0 \mid \bfX_n)} = \frac{\pi(k)}{\pi(k_0)} \prod_{j=2}^p \left(1+ \frac{1}{\gamma} \right)^{-\frac{k_j - k_{0j}}{2}} \left(\frac{\what{d}_{j}^{(k)} }{\what{d}_{j}^{(k_0)}} \right)^{-\frac{(1-\tau)n}{2}}.
	\eea
	Note that
	\bean\label{kneq_into_twoparts}
	\bbE_0 \Big[ \pi(k \neq k_0 \mid \bfX_n) \Big] &=& \sum_{k>k_0} \bbE_0 \Big[ \pi(k  \mid \bfX_n)  \Big]
	+ \sum_{k <k_0} \bbE_0 \Big[ \pi(k  \mid \bfX_n) \Big].
	\eean
	We will show that the right hand side terms in \eqref{kneq_into_twoparts} are of order $o(1)$.
	Because $\what{d}_j^{(k)}/ \what{d}_{j}^{(k_0)} \sim Beta ((n-k_j)/2, (k_j-k_{0j})/2)$ for any $k>k_0$, the first term in \eqref{kneq_into_twoparts} is bounded above by
	\bea
	&&\frac{\pi(k)}{\pi(k_0)} \prod_{j=2}^p \left(1+ \frac{1}{\gamma} \right)^{-\frac{k_j - k_{0j}}{2}} \cdot  \frac{\Gamma\big(\frac{n-k_{0j}}{2} \big) \Gamma\big(\frac{\tau n-k_{j}}{2} \big) }{\Gamma\big(\frac{n-k_{j}}{2} \big) \Gamma\big(\frac{\tau n-k_{0j}}{2} \big) } \\
	&\le& \frac{\pi(k)}{\pi(k_0)} \prod_{j=2}^p  \left(1+ \frac{1}{\gamma} \right)^{-\frac{k_j - k_{0j}}{2}} \left( \frac{1 - (k_{0j}+2)n^{-1} }{\tau - k_j n^{-1}} \right)^{\frac{k_j - k_{0j}}{2}}
	\eea 
	provided that $\tau n > R_n$.
	By condition (A3), we have $R_n/n \le \tau \epsilon_1/(1+\epsilon_1)$, 
	\bea
	\frac{1 - (k_{0j}+2)n^{-1} }{\tau - k_j n^{-1}} &\le& \frac{1+\epsilon_1}{\tau},
	\eea
	which implies 
	\bea
	\sum_{k>k_0} \bbE_0 \Big[ \pi(k  \mid \bfX_n)  \Big]
	&\le& \sum_{k>k_0}  \frac{\pi(k)}{\pi(k_0)} \left\{ \Big(1+ \frac{1}{\gamma} \Big)\cdot \Big( \frac{\tau}{1+\epsilon_1} \Big) \right\}^{- \frac{k-k_0}{2}\cdot (p - \frac{k+k_0+1}{2} ) }.
	\eea
	The last display is of order $o(1)$ by \eqref{size_cond}.
	%	Note that the desired property of this term mainly depends on $\pi(k)$ and the hyperparameters $\gamma$ and $\tau$.
	
	It is easy to check that
	\bean\label{ineqQjk}
	\left(\frac{\what{d}_{j}^{(k)} }{\what{d}_{j}^{(k_0)}} \right)^{-\frac{(1-\tau)n}{2}}   
	&\le& \exp \left( \frac{1-\tau}{2d_{0j}(1 + \what{Q}_{jk} )}\cdot n(\what{d}_j^{(k_0)} - \what{d}_j^{(k)}) \right) , 
	\eean
	where $\what{Q}_{jk} = \what{d}_{j}^{(k_0)}/d_{0j} -1 + (\what{d}_j^{(k)} - \what{d}_j^{(k_0)})/d_{0j}$.
	To deal with $\what{d}_j^{(k)}$ and $\what{Q}_{jk}$ easily, for a given constant $\epsilon = (\tau/10)^2$, we define the following sets 
	\bea
	N_j^c &=& \left\{ \bfX_n : \zeta_{0n}^{-1} (1-2\epsilon)^2 \le n^{-1}\lambda_{\min}(\bfX_{j(k_0)}^T \bfX_{j(k_0)}) \le n^{-1}\lambda_{\max}(\bfX_{j(k_0)}^T \bfX_{j(k_0)}) \le \epsilon_{0n}^{-1}(1+2\epsilon)^2  \right\}  , \\
	N_{1,j}^c &=& \left\{ \bfX_n: \bigg|\frac{\what{d}_j^{(k_0)}}{d_{0j}} -1 \bigg| \in \Big( -4 \sqrt{\epsilon} \frac{n-k_{0_j}}{n} - \frac{k_{0j}}{n}, 4\sqrt{\epsilon}\frac{n-k_{0_j}}{n} - \frac{k_{0j}}{n}  \Big)  \right\}, \\
	N_{2,j,k}^c &=& \left\{ \bfX_n : 0< \frac{\what{d}_j^{(k)} - \what{d}_j^{(k_0)} }{d_{0j}} < \epsilon + \frac{\what{\lambda}_{jk}}{n} \right\} 
	\eea
	and $N_{j,k}^c = N_j^c \cap N_{1,j}^c \cap N_{2,j,k}^c$, where $\what{\lambda}_{jk} = \|(I_n - \tilde{P}_{jk})\bfX_{j(k_0)} a_{0j}^{(k_0)}\|_2^2 / d_{0j}$.
	First, we will show that the above sets have probabilities tending to 1 as $n\to\infty$.
	Note that $n\what{d}_j^{(k_0)}/d_{0j} \sim \chi^2_{n-k_{0j}}$ and $n(\what{d}_j^{(k)} - \what{d}_j^{(k_0)})/d_{0j} \sim \chi^2_{k_{0j}-k_j}(\what{\lambda}_{jk})$, where $\chi^2_m(\lambda)$ denotes the noncentral chi-square distribution with degrees of freedom $m$ and the noncentrality parameter $\lambda$ and $\chi_m^2 = \chi_m^2(0)$.
	By Corollary 5.35 in \cite{eldar2012compressed}, $\bbP_0(N_j) \le 4 \exp(-n \epsilon^2/2)$ for all sufficiently large $n$.
	From the concentration inequality of chi-square random variable (Lemma 1 in \cite{laurent2000adaptive}), it is easy to see that $\bbP_0(N_{1,j}) \le 2 \exp(-\epsilon(n-k_0))$ for all sufficiently large $n$.
	Finally, by Lemma 4 in \cite{shin2018scalable}, we have 
	\bea
	\bbP_0 (N_{2,j,k}) &\lesssim&  \Big(\frac{\epsilon n}{2(k_{0j} - k_j)} \Big)^{\frac{k_{0j}-k_j}{2}} \exp\Big({\frac{k_{0j}-k_j}{2}- \frac{\epsilon n}{2} }\Big) + \bbE_0 \bigg(\, \frac{\what{\lambda}_{jk}}{\epsilon n} e^{-\frac{\epsilon^2 n^2}{32 \what{\lambda}_{jk} }}  \wedge 1 \,\bigg) \\
	&\le& \exp\Big({-\frac{\epsilon n}{4}}\Big) + \bbE_0 \bigg[\, \frac{\what{\lambda}_{jk}}{\epsilon n} \exp\bigg({-\frac{\epsilon^2 n^2}{32 \what{\lambda}_{jk} }}\bigg) \cdot I(N_{j}^c)\,\bigg] + \bbP_0(N_j) \\
	&\le& \exp\Big({-\frac{\epsilon n}{4}}\Big)  +   \exp\Big({ - \frac{\epsilon^2 \, \zeta_{0n}^{-1} \epsilon_{0n} n}{128 (1+2\epsilon)^2 } }\Big)    + 4\exp\Big({-\frac{\epsilon^2 n}{2}}\Big) ,
	\eea
	which is of order $o(1)$ provided that $\zeta_{0n}/\epsilon_{0n} = o(n)$.
	The last inequality holds because 
	\bea
	\what{\lambda}_{jk} &\le& \lambda_{\max}(\bfX_{j(k_0)}^T \bfX_{j(k_0)}) \cdot d_{0j}^{-1}  \|a_{0j}\|_2^2 \\
	&\le& 2n (1+2\epsilon)^2 \epsilon_{0n}^{-1} \cdot \Big\{ \| d_{0j}^{-1/2}(e_j - a_{0j})\|_2^2 + d_{0j}^{-1}  \Big\} \\
	&\le& 2n (1+2\epsilon)^2 \epsilon_{0n}^{-1} \cdot \zeta_{0n}
	\eea 
	on $N_j^c$ for all sufficiently large $n$, where $e_j$ is the unit vector whose $j$th element is 1 and the others are zero.
	Note that
	\bean
	\sum_{k <k_0} \bbE_0 \Big[ \pi(k  \mid \bfX_n) \Big]
	&\le& \sum_{k <k_0} \sum_{j=2}^p \bbP_0 (N_{j,k}) \label{P0_Nsets} + \sum_{k <k_0} \bbE_0 \Big[\pi(k\mid \bfX_n) \prod_{j=2}^p I( N_{j,k}^c) \Big]. \label{P0_Ncomp}
	\eean
	By the above arguments, \eqref{P0_Nsets} is of order $o(1)$ provided that $\zeta_{0n} \log p/\epsilon_{0n} = o(n)$.
	Thus, the proof is completed if we prove that \eqref{P0_Ncomp} is of order $o(1)$.
	
	From the inequality \eqref{ineqQjk}, 
	\bea
	&& \sum_{k <k_0} \bbE_0 \Big[\pi(k\mid \bfX_n) \prod_{j=2}^p I( N_{j,k}^c) \Big] \\
	&\le& \sum_{k <k_0} \bbE_0 \left[ \frac{\pi(k)}{\pi(k_0)} \prod_{j=2}^p \Big(1+\frac{1}{\gamma} \Big)^{\frac{k_{0j}-k_j}{2}} \exp\Big( \frac{1-\tau}{2d_{0j}(1+\what{Q}_{jk}) }\cdot n(\what{d}_j^{(k_0)} - \what{d}_j^{(k)}) \Big)  I( N_{j,k}^c) \right].
	\eea
	On the event $N_{j,k}^c$, we have
	\bea
	\what{Q}_{jk} &\le& 4\sqrt{\epsilon}\frac{n-k_0 }{n} - \frac{k_0}{n} + \epsilon + \frac{\what{\lambda}_{jk}}{n} \\
	&\le& 5 \sqrt{\epsilon} + \frac{\what{\lambda}_{jk}}{n} ,\\
	\what{Q}_{jk} &\ge& -4\sqrt{\epsilon} \frac{n-k_0 }{n} - \frac{k_0}{n} \,\,\ge\,\, -5 \sqrt{\epsilon}
	\eea
	for all sufficiently large $n$.
	For a given $k <k_0$,
	\bea
	n( \what{d}_j^{(k_0)} - \what{d}_j^{(k)} )
	&=& \tilde{X}_j^T (\tilde{P}_{jk} - \tilde{P}_{jk_0}) \tilde{X}_j \\
	&\overset{d}{\equiv}& - \|(I_n- \tilde{P}_{jk})\bfX_{j(k_0)} a_{0j}^{(k_0)}\|_2^2 - 2 \tilde{\epsilon}_j^T (I_n - \tilde{P}_{jk})\bfX_{j(k_0)} a_{0j}^{(k_0)} + \tilde{\epsilon}_j^T (\tilde{P}_{jk} - \tilde{P}_{jk_0}) \tilde{\epsilon}_j \\
	&\le& -d_{0j}\what{\lambda}_{jk} - 2 \tilde{\epsilon}_j^T (I_n - \tilde{P}_{jk})\bfX_{j(k_0)} a_{0j}^{(k_0)} \\
	&=:& -d_{0j}\what{\lambda}_{jk} - 2 V_{jk},
	\eea
	where $\tilde{\epsilon}_j \sim N_n(0, d_{0j}I_n)$ and $V_{jk}/\sqrt{d_{0j}} \sim N(0, d_{0j}\what{\lambda}_{jk} )$ under $\bbP_0$ given $\bfX_{j(k_0)}$.
	Then,
	\bea
	&& \bbE_0 \left[ \exp\Big( \frac{1-\tau}{2d_{0j}(1+\what{Q}_{jk}) }\cdot n(\what{d}_j^{(k_0)} - \what{d}_j^{(k)}) \Big)  \mid \bfX_{j(k_0)} \right] I(N_{j,k}^c) \\
	&\le& \bbE_0 \left[ \exp\Big( -\frac{1-\tau}{2d_{0j}(1+\what{Q}_{jk}) }\cdot  (d_{0j}\what{\lambda}_{jk} +2 V_{jk}) \Big)  \mid \bfX_{j(k_0)} \right] I(N_{j,k}^c) .
	\eea
	From the moment generating function of the normal distribution, we have
	\bea
	&& \bbE_0 \left[ \exp\Big( -\frac{1-\tau}{2d_{0j}(1 + Q) }\cdot  (d_{0j}\what{\lambda}_{jk} +2 V_{jk}) \Big)  \mid \bfX_{j(k_0)} \right] I(N_{j,k}^c)\\
	&\le& \exp \Big\{ -\frac{1-\tau}{2(1+Q) }\Big(1 - \frac{1-\tau}{1+Q} \Big)\, \what{\lambda}_{jk}  \Big\} I(N_{j,k}^c) \\
	&\le& \exp \Big\{ -\frac{1-\tau}{2( 1+ 5\sqrt{\epsilon} + \what{\lambda}_{jk}/n ) }\Big(1 - \frac{1-\tau}{1- 5\sqrt{\epsilon}} \Big)\, \what{\lambda}_{jk} \Big\}  I(N_{j,k}^c)
	\eea
	for any $Q=-5\sqrt{\epsilon}$ or $Q=5\sqrt{\epsilon} + \what{\lambda}_{jk}/n$. 
	Note that
	\bea
	d_{0j}\what{\lambda}_{jk}&=& \|(I_n - \tilde{P}_{jk})\bfX_{j(k_0)} a_{0j}^{(k_0)}\|_2^2 \\
	&\ge& \lambda_{\min}(\bfX_{j(k_0)}^T \bfX_{j(k_0)}) (k_{0j} -k_j) \min_{j,l: a_{0,jl}\neq 0}|a_{0,jl}|^2 \\
	&\ge& n \zeta_{0n}^{-1}(1-2\epsilon)^2 (k_{0j} -k_j) \min_{j,l: a_{0,jl}\neq 0}|a_{0,jl}|^2
	\eea
	on $\bfX_n \in N_{j,k}^c$ by Lemma 5 in \cite{arias2014estimation} and $1-(1-\tau)/(1-5\sqrt{\epsilon}) > \tau/2$ by the definition of $\epsilon$.
	Thus, 
	\bea
	&& \frac{1-\tau}{2( 1+ 5\sqrt{\epsilon} + \what{\lambda}_{jk}/n ) }\Big(1 - \frac{1-\tau}{1- 5\sqrt{\epsilon}} \Big)\, \what{\lambda}_{jk} \\
	&\ge& \frac{\tau(1-\tau)}{4} \cdot \Big( \frac{1+5\sqrt{\epsilon}}{\what{\lambda}_{jk}} + \frac{1}{n}  \Big)^{-1}  \\
	&\ge& \frac{\tau(1-\tau)}{4} \cdot \Big( \frac{(1+5\sqrt{\epsilon})\epsilon_{0n}^{-1}\zeta_{0n} }{(1-2\epsilon)^2 (k_{0j}-k_j) \min_{j,l: a_{0,jl}\neq 0}|a_{0,jl}|^2 n  } + \frac{1}{n}  \Big)^{-1} \\
	&\ge& \frac{\tau(1-\tau)}{4} \cdot \Big( \frac{(1+5\sqrt{\epsilon}) }{(1-2\epsilon)^2 (k_{0j}-k_j) C_{\rm bm}} + \frac{1}{n}  \Big)^{-1}  \\
	&\ge& \frac{\tau(1-\tau)}{8} \cdot  \frac{(1-2\epsilon)^2 (k_{0j}-k_j) C_{\rm bm}  }{(1+5\sqrt{\epsilon})}  \\
	&\ge& (k_{0j}-k_j) M_{\rm bm} 
	\eea
	on $\bfX_n \in N_{j,k}^c$ by  condition (A2), the definition of $\epsilon$ and $\tau \le 0.4$, where $C_{\rm bm} = 10\tau^{-1}(1-\tau)^{-1} M_{\rm bm}$.
	It implies that \eqref{P0_Ncomp} is bounded above by
	\bea
	&& \sum_{k <k_0} \frac{\pi(k)}{\pi(k_0)} \prod_{j=2}^p \Big(1+\frac{1}{\gamma}\Big)^{\frac{k_{0j}-k_j}{2}} \cdot 2 (e^{-M_{\rm bm}})^{ k_{0j}-k_j } \\
	&\le& \sum_{k <k_0} \frac{\pi(k)}{\pi(k_0)} \prod_{j=2}^p  \left\{ 2\Big(1+\frac{1}{\gamma}\Big)^{\frac{1}{2}} e^{-M_{\rm bm}} \right\}^{ k_{0j}-k_j } \\
	&=& \sum_{k <k_0} \frac{\pi(k)}{\pi(k_0)}  \left\{ 2\Big(1+\frac{1}{\gamma}\Big)^{\frac{1}{2}} e^{-M_{\rm bm}} \right\}^{ (k_{0}-k) \big(p - \frac{k+k_0+1}{2} \big) } ,
	\eea
	which is of order $o(1)$ provided that \eqref{size_cond2}.

\end{proof}

\begin{proof}[Proof of Theorem \ref{theorem:one_band}]
	Note that
	\bea
	B_{10}(\bfX_n) &=& \frac{p(\bfX_n \mid H_1)}{p(\bfX_n \mid H_0)}  \\
	&=& \frac{ \sum_{k > k^\star} \int p(\bfX_n \mid \Omega_n, k) \pi (\Omega_n \mid k)\pi_1(k) d\Omega_n }{\sum_{k \le k^\star} \int p(\bfX_n \mid \Omega_n, k) \pi (\Omega_n\mid k) \pi_0(k) d\Omega_n}  \\
	&=& \frac{ \sum_{k > k^\star}  p(\bfX_n \mid k) \pi_1(k)  }{\sum_{k \le k^\star} \int p(\bfX_n \mid k)  \pi_0(k) } \\
	&=& \frac{\pi(k >k^\star\mid \bfX_n)}{\pi(k \le k^\star \mid \bfX_n)}\times \frac{C_0}{C_1}
	\eea
	and $C_0/C_1 = k^\star / (R_n- k^\star)$.
	
	If $H_0: k \le k^\star$ is true, i.e. $k_0 \le k^\star$,
	\bea
	\bbE_0 \Big[ B_{10}(\bfX_n)\Big] &\le& \bbE_0 \Big[ \frac{\pi(k > k^\star \mid \bfX_n)}{\pi(k = k_0\mid \bfX_n)} \Big] \times \frac{k^\star}{R_n- k^\star} \\
	&\le& \sum_{k > k^\star} \bbE_0 \Big[ \frac{\pi(k \mid \bfX_n) }{\pi(k_0\mid \bfX_n)}\Big] \times \frac{k^\star}{R_n- k^\star} \\
	&\le& k^\star \cdot \left\{ \frac{(1+\gamma)\tau}{\gamma(1+\tau)} \right\}^{- \frac{k^\star+1-k_0}{2}\cdot (p - \frac{R_n+k_0+1}{2} ) } = T_{n,H_0,k_0,k^\star},
	\eea
	which implies
	\bea
	B_{10}(\bfX_n)
	&=& O_p \big( T_{n,H_0,k_0,k^\star} \big) \quad\text{ under } H_0.
	\eea
	
	On the other hand, if $H_1: k > k^\star$ is true, i.e. $k_0 > k^\star$, 
	\bea
	&& \bbE_0 \Big[ \frac{\sum_{k \le k^\star} \pi(k\mid \bfX_n) \prod_{j=2}^p I(N_{j,k}^c)}{\pi(k_0 \mid \bfX_n)} \Big] \times \frac{R_n-k^\star}{k^\star} \\
	&\le& \left\{ \frac{4(1+\gamma)}{\gamma \exp(2M_{\rm bm}) }  \right\}^{\frac{k_0-k^\star}{2}\cdot (p - \frac{k^\star+k_0+1}{2} ) } \times (R_n-k^\star) \\
	&=& T_{n,H_1,k_0,k^\star} .
	\eea 
	Now, we will show that for every $\epsilon>0$, there exist a constant $C>0$ and an integer $N$ such that 
	\bea
	\bbP_0 \Big( B_{10}(\bfX_n) < C^{-1} T_{n,H_1,k_0,k^\star}^{-1}  \Big)
	&\le& \epsilon
	\eea
	for all $n \ge N$ under $H_1$, which implies $B_{10}^{-1}(\bfX_n) = O_p(T_{n,H_1,k_0,k^\star})$ under $H_1$.
	Note that 
	\bea
	&& \bbP_0 \Big( B_{10}(\bfX_n) \le C^{-1} T_{n,H_1,k_0,k^\star}^{-1}  \Big) \\
	&=& \bbP_0 \Big( B_{10}(\bfX_n)^{-1} \ge C T_{n,H_1,k_0,k^\star}  \Big) \\
	&\le& \bbP_0 \Big( \frac{\sum_{k \le k^\star} \pi(k\mid \bfX_n) }{\pi(k_0 \mid \bfX_n) }  \ge  \frac{k^\star C T_{n,H_1,k_0,k^\star}}{R_n-k^\star} \Big) .
	\eea
	Let $N_{j,k}$ be the set defined in the proof of Theorem \ref{theorem:band_sel},	then the last term is bounded above by
	\bea
	&& \bbP_0 \Big( \frac{ \sum_{k \le k^\star}\pi(k\mid \bfX_n)\prod_{j=2}^p I(N_{j,k}^c) }{\pi(k_0 \mid \bfX_n) }   \ge  \frac{k^\star C T_{n,H_1,k_0,k^\star} }{R_n-k^\star}  \Big) + \sum_{k \le k^\star} \sum_{j=2}^p \bbP_0(N_{j,k}) \\
	&\le& \bbE_0 \Big[ \frac{\sum_{k \le k^\star} \pi(k\mid \bfX_n) \prod_{j=2}^p I(N_{j,k}^c)}{\pi(k_0 \mid \bfX_n)} \Big] \times \frac{R_n- k^\star}{k^\star C T_{n,H_1,k_0,k^\star} } + o(1) \\
	&\le& \frac{1}{C} + o(1).
	\eea
	Thus, it completes the proof. 
\end{proof}

\begin{proof}[Proof of Theorem \ref{theorem:two_band}]
	It is easy to see that
	\bea
	B_{10}(\bfX_{n_1}, \bfY_{n_2}) &=& \frac{p(\bfX_{n_1}, \bfY_{n_2} \mid H_1)}{p(\bfX_{n_1}, \bfY_{n_2} \mid H_0)}\\
	&=& \frac{\sum_{k_1\neq k_2} \pi(k_1\mid \bfX_{n_1})\pi(k_2\mid \bfY_{n_2})  }{\sum_{k_1= k_2} \pi(k_1\mid \bfX_{n_1})\pi(k_2\mid \bfY_{n_2})} \times R_n^{-1} .
	\eea
	
	If $H_0: k_1=k_2$ is true, let $k_{01}=k_{02}=k_0$, then
	\bea
	B_{10}(\bfX_{n_1}, \bfY_{n_2})
	&\le& \frac{\sum_{k_1\neq k_2} \pi(k_1\mid \bfX_{n_1})\pi(k_2\mid \bfY_{n_2})  }{ \pi(k_0\mid \bfX_{n_1})\pi(k_0\mid \bfY_{n_2})} \times R_n^{-1} \\
	&\le& \sum_{k_1=1}^{R_n} \frac{\pi(k_1\mid \bfX_{n_1})}{\pi(k_0\mid \bfX_{n_1})} \max_{k_2\neq k_1} \frac{\pi(k_2\mid \bfY_{n_2})}{\pi(k_0\mid \bfY_{n_2})} .
	\eea
	Note that
	\bea
	\frac{\pi(k_0\mid \bfX_{n_1})}{\pi(k_0\mid \bfX_{n_1})} \max_{k_2\neq k_0} \frac{\pi(k_2\mid \bfY_{n_2})}{\pi(k_0\mid \bfY_{n_2})} 
	&=& \max_{k_2\neq k_0} \frac{\pi(k_2\mid \bfY_{n_2})}{\pi(k_0\mid \bfY_{n_2})}  \\
	&=& O_p \left( \frac{1}{R_n-k_{0}} T_{n, H_1, k_0, k_0-1} + \frac{1}{k_0} T_{n, H_0, k_0,k_0}    \right) 
	\eea
	and
	\bea
	&& \sum_{k_1\neq k_0} \frac{\pi(k_1\mid \bfX_{n_1})}{\pi(k_0\mid \bfX_{n_1})} \max_{k_2\neq k_1} \frac{\pi(k_2\mid \bfY_{n_2})}{\pi(k_0\mid \bfY_{n_2})} \\
	&=& O_p \left( \frac{k_0}{R_n-k_{0}} T_{n,H_1,k_0, k_0-1} + \frac{R_n-k_0}{k_0} T_{n,H_0, k_0, k_0}   \right).
	\eea
	Thus, we have under $H_0$,
	\bea
	B_{10}(\bfX_{n_1}, \bfY_{n_2})
	&=&  O_p \left( \frac{k_0}{R_n-k_{0}} T_{n,H_1,k_0, k_0-1} + \frac{R_n-k_0}{k_0} T_{n,H_0, k_0, k_0}   \right).
	\eea
	%under $H_0$.
	
	On the other hand, if $H_1:k_1\neq k_2$ is true, 
	\bea
	B_{10}(\bfX_{n_1}, \bfY_{n_2})^{-1}
	&\le& \frac{\sum_{k_1= k_2} \pi(k_1\mid \bfX_{n_1})\pi(k_2\mid \bfY_{n_2})}{\pi(k_{01}\mid \bfX_{n_1})\pi(k_{02}\mid \bfY_{n_2})} \times R_n.
	\eea
	Note that 
	\bea
	\frac{\pi(k_{01}\mid \bfX_{n_1})\pi(k_{01}\mid \bfY_{n_2})}{\pi(k_{01}\mid \bfX_{n_1})\pi(k_{02}\mid \bfY_{n_2})} \times R_n 
	&=& \frac{ \pi(k_{01}\mid \bfY_{n_2})}{ \pi(k_{02}\mid \bfY_{n_2})} \times R_n  \\
	&=& O_p \left( \frac{R_n}{R_n-k_{01}} T_{n, H_1, k_{02}, k_{01} } + \frac{R_n}{k_{01}} T_{n, H_0, k_{02}, k_{01}}  \right)
	\eea
	and
	\bea
	&& \frac{\sum_{k_1= k_2 \neq k_{01}} \pi(k_1\mid \bfX_{n_1})\pi(k_2\mid \bfY_{n_2})}{\pi(k_{01}\mid \bfX_{n_1})\pi(k_{02}\mid \bfY_{n_2})} \times R_n \\
	&\le& \sum_{k_1 \neq k_{01}} \frac{\pi(k_1\mid \bfX_{n_1})}{\pi(k_{01}\mid \bfX_{n_1})} \max_{k_2} \frac{\pi(k_2\mid \bfY_{n_1})}{\pi(k_{02}\mid \bfY_{n_1})} \times R_n \\
	&=& O_p \left( \frac{R_n \, k_{01}}{R_n-k_{01}} T_{n,H_1, k_{01}, k_{01}-1} + \frac{R_n(R_n-k_{01})}{k_{01}} T_{n,H_0, k_{01},k_{01}}   \right) .
	\eea
	Thus, we have
	\bea
	B_{10}(\bfX_{n_1}, \bfY_{n_2})^{-1}
	&=& O_p \left(  \frac{R_n \, k_{01}}{R_n-k_{01}} T_{n,H_1, k_{01}, k_{01}-1} + \frac{R_n(R_n-k_{01})}{k_{01}} T_{n,H_0, k_{01},k_{01}} \right).
	\eea
	Similarly, it is easy to show that
	\bea
	B_{10}(\bfX_{n_1}, \bfY_{n_2})^{-1}
	&=& O_p \left(  \frac{R_n \, k_{02}}{R_n-k_{02}} T_{n,H_1, k_{02}, k_{02}-1} + \frac{R_n(R_n-k_{02})}{k_{02}} T_{n,H_0, k_{02},k_{02}} \right).
	\eea
	Let $k_{\min} = k_{01}\wedge k_{02}$, then under $H_1$ one has
	\bea
	B_{10}(\bfX_{n_1}, \bfY_{n_2})^{-1}
	&=& O_p \left(  \frac{R_n \, k_{\min}}{R_n-k_{\min}} T_{n,H_1, k_{\min}, k_{\min}-1} + \frac{R_n(R_n-k_{\min})}{k_{\min}} T_{n,H_0, k_{\min},k_{\min}} \right).
	\eea
	%under $H_1$.
\end{proof}

\begin{proof}[Proof of Corollary \ref{cor:Ploss_rate}]
	By the proof of Theorem \ref{theorem:band_sel}, we have
	\bea
	\bbE_0 \big[ \pi(k\neq k_0 \mid \bfX_n) \big] &\lesssim& C_{\star}^{-(p-k_0-1)},
	\eea
	where $C_{\star}>0$ is a constant depending on $\gamma, \tau$ and $M_{\rm bm}$.
	By the Markov's inequality, 
	\bean
	\bbP_0 (\what{k} \neq k_0)  &\le&  \bbP_0( \pi(k\neq k_0 \mid \bfX_n) > 1/2 )  \nonumber   \\
	&\le& 2\, \bbE_0 \big[ \pi(k\neq k_0 \mid \bfX_n) \big], \label{P0khat}
	\eean
	where $\what{k} = \argmax_k \pi(k \mid \bfX_n)$.
	
	Note that
	\bean
	&& \bbE_0 \bbE^{\pi_{(k)}} \Big( \| A_{n} - A_{0n} \|_\infty  \mid \bfX_n \Big)  \nonumber\\
	&=& \bbE_0 \Big[ \bbE^{\pi_{(k)}} \Big( \| A_{n} - A_{0n} \|_\infty  \mid \bfX_n \Big) I(\what{k} = k_0) \Big] + \bbE_0 \Big[ \bbE^{\pi_{(k)}} \Big( \| A_{n} - A_{0n} \|_\infty  \mid \bfX_n \Big) I(\what{k} \neq k_0) \Big]  . \quad\quad \label{Prisk_decomp}
	\eean
	The first term in \eqref{Prisk_decomp} is of order $( k_0(k_0 + \log p)/n )^{1/2}$ by Lemmas 2 and 4 in \cite{lee2017estimating}.
	The second term in \eqref{Prisk_decomp} is bounded above by
	\bea
	&& \Big\{ \bbE_0 \bbE^{\pi_{(k)}} \big( \|A_{n} - A_{0n} \|_\infty^2 \mid \bfX_n \big) \Big\}^{\frac{1}{2}} \, \left\{  \bbP_0 (\what{k} \neq k_0) \right\}^{\frac{1}{2}}  	\lesssim    \left( \frac{k_0 ( k_0 + \log p)}{n} \right)^{\frac{1}{2}}
	%	&\le& \left\{  \sum_{k\neq k_0}^{R_n} \bbE_0 \bbE^\pi \Big( \| A_{n}^{(k)} - A_{0n} \|_\infty^2  \mid \bfX_n \Big) \right\}^{\frac{1}{2}}  \, \left\{  \bbP_0 (\what{k} \neq k_0)\right\}^{\frac{1}{2}}  \\
%	&\lesssim &   \left( \frac{k_0 ( k_0 + \log p)}{n} \right)^{\frac{1}{2}}
	\eea
	by \eqref{P0khat} and Lemmas 2 and 4 in \cite{lee2017estimating}.
\end{proof}

\bibliographystyle{dcu}
\bibliography{bandwidth-selection-chol}

\end{document}